\documentclass[11pt]{article}

\pdfoutput=1

\usepackage{a4}
\usepackage[top=30truemm,bottom=30truemm]{geometry}

\usepackage{authblk}
\usepackage{ascmac}
\usepackage{amsmath,amssymb}
\usepackage{amsthm}
\usepackage{tabls}
\usepackage{graphicx}
\usepackage{wrapfig}
\usepackage{array}
\usepackage{url}
\usepackage{multirow}
\usepackage{multicol}
\usepackage{cases}
\usepackage[normalem]{ulem}
\usepackage{boites}
\usepackage[final,mode=multiuser]{fixme}
\fxuseenvlayout{colorsig}
\fxusetargetlayout{color}
\FXRegisterAuthor{an}{ann}{AN}
\FXRegisterAuthor{nf}{a}{NF}
\FXRegisterAuthor{yn}{ayn}{YN}
\FXRegisterAuthor{si}{asi}{SI}
\fxsetface{env}{}

\newtheorem{theorem}{Theorem}
\newtheorem{lemma}{Lemma}
\newtheorem{corollary}{Corollary}

\theoremstyle{definition}
\newtheorem{problem}{Problem}

\newcommand{\Prefix}{\mathsf{Prefix}}

\newcommand{\SHT}{\mathsf{SHT}}
\newcommand{\CPH}{\mathsf{CPH}}

\newcommand{\parent}{\mathsf{parent}}

\newcommand{\anc}{\mathsf{anc}}
\newcommand{\nma}{\mathsf{nma}}

\newcommand{\depth}{\mathsf{depth}}
\newcommand{\str}{\mathsf{str}}

\newcommand{\slink}{\mathsf{sl}}
\newcommand{\rslink}{\mathsf{rsl}}
\newcommand{\mrp}{\mathsf{mrp}}

\newcommand{\occ}{\mathit{occ}}

\newcommand{\inputtrie}{\boldsymbol{T}}
\newcommand{\fptrie}{\boldsymbol{T}_{\mathsf{FP}}}

\newcommand{\na}{\mathsf{na}}

\newcommand{\TrieSigma}{\Sigma_{\inputtrie}}
\newcommand{\triesigma}{\sigma_{\inputtrie}}

\newcommand{\CT}{\mathsf{CT}}
\newcommand{\PD}{\mathsf{PD}}
\newcommand{\FP}{\mathsf{FP}}

\newcommand{\fact}{\mathsf{f}}

\newcommand{\lensum}{\mathsf{lsum}}

\newcommand{\Front}{\mathcal{F}}
\newcommand{\Zero}{\mathcal{Z}}

\title{Position Heaps for Cartesian-tree Matching \\ on Strings and Tries}

\author{Akio~Nishimoto$^{1}$}
\author{Noriki~Fujisato$^{1}$}
\author{Yuto~Nakashima$^{1}$}
\author{Shunsuke~Inenaga$^{1,2}$}

\affil{
  \normalsize{
  \textit{$^1$Department of Informatics, Kyushu University, Japan}\\
  \texttt{\{nishimoto.akio, noriki.fujisato, yuto.nakashima, inenaga\}@inf.kyushu-u.ac.jp}\\
  \textit{$^2$PRESTO, Japan Science and Technology Agency, Japan}
  }
}

\date{}

\begin{document}
\maketitle

\begin{abstract}
  The \emph{Cartesian-tree pattern matching} is a recently introduced scheme
  of pattern matching that detects fragments in a sequential data stream
  which have a similar structure as a query pattern.
  Formally, Cartesian-tree pattern matching seeks all substrings $S'$ of
  the text string $S$ such that the Cartesian tree of $S'$
  and that of a query pattern $P$ coincide.
  In this paper, we present a new indexing structure for this problem,
  called the \emph{Cartesian-tree Position Heap} (\emph{CPH}).
  Let $n$ be the length of the input text string $S$,
  $m$ the length of a query pattern $P$, and $\sigma$ the alphabet size.
  We show that the CPH of $S$, denoted $\CPH(S)$,
  supports pattern matching queries in
  $O(m (\sigma + \log (\min\{h, m\})) + \occ)$ time with $O(n)$ space,
  where $h$ is the height of the CPH and $\occ$ is the number of
  pattern occurrences.
  We show how to build $\CPH(S)$ in $O(n \log \sigma)$ time with $O(n)$ working space.
  Further, we extend the problem to the case where the text is a
  labeled tree (i.e. a trie).
  Given a trie $\inputtrie$ with $N$ nodes,
  we show that the CPH of $\inputtrie$, denoted $\CPH(\inputtrie)$,
  supports pattern matching queries on the trie in
  $O(m (\sigma^2 + \log (\min\{h, m\})) + \occ)$ time with $O(N \sigma)$ space.
  We also show a construction algorithm for $\CPH(\inputtrie)$
  running in $O(N \sigma)$ time and $O(N \sigma)$ working space.
\end{abstract}

\section{Introduction} \label{sec:intro}

%\subsection{Cartesian-tree Pattern Matching}

If the Cartesian trees $\CT(X)$ and $\CT(Y)$
of two strings $X$ and $Y$ are equal,
then we say that $X$ and $Y$ \emph{Cartesian-tree match} (\emph{ct-match}).
The \emph{Cartesian-tree pattern matching problem} (\emph{ct-matching problem})
~\cite{ParkBALP20} is, given a text string $S$ and a pattern $P$,
to find all substrings $S'$ of $S$ that ct-match with $P$.

String equivalence with ct-matching 
belongs to the class of \emph{substring-consistent equivalence relation} (\emph{SCER})~\cite{MatsuokaAIBT16},
namely, the following holds:
If two strings $X$ and $Y$ ct-match,
then $X[i..j]$ and $Y[i..j]$ also ct-match for any $1 \leq i \leq j \leq |X|$.
Among other types of SCERs (\cite{Baker93,baker95parameterized,Baker96,IIT11,KimH16}), ct-matching is the most related to order-peserving matching
(op-matching)~\cite{KimEFHIPPT14,ChoNPS15,CrochemoreIKKLP16}.
Two strings $X$ and $Y$ are said to op-match
if the relative order of the characters in $X$
and the relative order of the characters in $Y$ are the same.
It is known that with ct-matching
one can detect some interesting occurrences of a pattern
that cannot be captured with op-matching.
More precisely, if two strings $X$ and $Y$ op-match,
then $X$ and $Y$ also ct-match. However, the reverse is not true.
With this property in hand, ct-matching is motivated for analysis of
time series such as stock charts~\cite{ParkBALP20,FuCLN07}.

This paper deals with the indexing version of the ct-matching problem.
Park et al.~\cite{ParkBALP20} proposed the \emph{Cartesian suffix tree}
(\emph{CST}) for a text string $S$ that can be built in
$O(n \log n)$ worst-case time or $O(n)$ expected time,
where $n$ is the length of the text string $S$.
The $\log n$ factor in the worst-case complexity
is due to the fact that the \emph{parent-encoding},
a key concept for ct-matching introduced in~\cite{ParkBALP20},
is a sequence of integers in range $[0..n-1]$.
While it is not explicitly stated in Park et al.'s paper~\cite{ParkBALP20},
our simple analysis (c.f. Lemma~\ref{lem:number_of_branches} in Section~\ref{sec:matching})
reveals that the CST supports pattern matching queries
in $O(m \log m + \occ)$ time,
where $m$ is the pattern length and $\occ$ is the number of
pattern occurrences.

%\subsection{Our Contribution}

In this paper, we present a new indexing structure for this problem,
called the \emph{Cartesian-tree Position Heap} (\emph{CPH}).
%Let $\sigma_S$ be the number of distinct characters in the input string $S$.
We show that the CPH of $S$, which occupies $O(n)$ space,
can be built in $O(n \log \sigma)$ time with $O(n)$ working space
and supports pattern matching queries in
$O(m (\sigma + \log (\min\{h, m\})) + \occ)$ time,
where $h$ is the height of the CPH.
Compared to the afore-mentioned CST,
our CPH is the \emph{first} index for ct-matching
that can be built in worst-case linear time for constant-size alphabets,
while pattern matching queries with our CPH can be slower
than with the CST when $\sigma$ is large.

We then consider the case where the text is a labeled tree (i.e. a trie).
Given a trie $\inputtrie$ with $N$ nodes,
we show that the CPH of $\inputtrie$,
which occupies $O(N \sigma)$ space,
can be built in $O(N \sigma)$ time and $O(N \sigma)$ working space.
We also show how to support pattern matching queries
in $O(m (\sigma^2 + \log (\min\{h, m\})) + \occ)$ time in the trie case.
To our knowledge,
our CPH is the first indexing structure for ct-matching on tries
that uses linear space for constant-size alphabets.

Conceptually, our CPH is most related to
the \emph{parameterized position heap} (\emph{PPH}) for
a string~\cite{FujisatoNIBT18} and for a trie~\cite{FujisatoNIBT19b},
in that our CPHs and the PPHs are both constructed
in an incremental manner where the suffixes of an input string and
the suffixes of an input trie
are processed in increasing order of their lengths.
However, some new techniques are required in the construction of
our CPH due to different nature of the
\emph{parent encoding}~\cite{ParkBALP20} of strings for ct-matching,
from the \emph{previous encoding}~\cite{Baker93} of strings for
parameterized matching.

\section{Preliminaries}

\subsection{Strings and (Reversed) Tries}

Let $\Sigma$ be an ordered alphabet of size $\sigma$.
An element of $\Sigma$ is called a \emph{character}.
An element of $\Sigma^*$ is called a \emph{string}.
For a string $S \in \Sigma^*$,
let $\sigma_S$ denote the number of distinct characters in $S$.

The empty string $\varepsilon$ is a string of length 0,
namely, $|\varepsilon| = 0$.
For a string $S = XYZ$, $X$, $Y$ and $Z$ are called
a \emph{prefix}, \emph{substring}, and \emph{suffix} of $S$, respectively.
The set of prefixes of a string $S$ is denoted by $\Prefix(S)$.
%and the set of suffixes of $S$ is denoted by $\Suffix(S)$.
The $i$-th character of a string $S$ is denoted by
$S[i]$ for $1 \leq i \leq |S|$,
and the substring of a string $S$ that begins at position $i$ and
ends at position $j$ is denoted by $S[i..j]$ for $1 \leq i \leq j \leq |S|$.
For convenience, let $S[i..j] = \varepsilon$ if $j < i$.
Also, let $S[i..] = S[i..|S|]$ for any $1 \leq i \leq |S|+1$.
%For any string $w$, let $w^R$ denote the reversed string of $w$, i.e., $w^R = w[|w|] \cdots w[1]$.
%For any (p-)string $S$, let $\rev{S}$ denote the reversed string of $S$,
%i.e., $\rev{S} = S[|S|]S[|S|-1] \cdots S[1]$.

A \emph{trie} is a rooted tree that represents a set of strings,
where each edge is labeled with a character from $\Sigma$
and the labels of the out-going edges of each node is mutually distinct.
Tries are natural generalizations to strings in that
tries can have branches while strings are sequences without branches.

Let $\mathbf{x}$ be any node of a given trie $\inputtrie$,
and let $\mathbf{r}$ denote the root of $\inputtrie$.
Let $\depth(\mathbf{x})$ denote the depth of $\mathbf{x}$.
When $\mathbf{x} \neq \mathbf{r}$,
let $\parent(\mathbf{x})$ denote the parent of $\mathbf{x}$.
For any $0 \leq j \leq \depth(\mathbf{x})$,
let $\anc(\mathbf{x}, j)$ denote the $j$-th ancestor of $\mathbf{x}$,
namely, $\anc(\mathbf{x}, 0) = \mathbf{x}$ and $\anc(\mathbf{x}, j) = \parent(\anc(\mathbf{x}, j-1))$
for $1 \leq j \leq \depth(\mathbf{x})$.
It is known that after a linear-time processing on $\inputtrie$,
$\anc(\mathbf{x}, j)$ for any query node $\mathbf{x}$ and integer $j$
can be answered in $O(1)$ time~\cite{BenderF04}.

\sinote*{modified}{%
For the sake of convenience, 
in the case where our input is a trie $\inputtrie$,
then we consider its \emph{reversed trie} where the path labels are read
in the leaf-to-root direction.
On the other hand, the trie-based data structures
(namely position heaps) we build for input strings and reversed tries
are usual tries where the path labels are read
in the root-to-leaf direction.
}%
%In the sequel, when we talk about tries,
%we mean reversed tries, unless otherwise stated.

%We will sometimes identify each node $v$ of $\inputtrie$
%with the path string that is obtained by concatenating
%the edge labels from $v$ to the root of $\inputtrie$.
For each (reversed) path $(\mathbf{x}, \mathbf{y})$ in $\inputtrie$
such that $\mathbf{y} = \anc(\mathbf{x}, j)$
with $j = |\depth(\mathbf{x})|-|\depth(\mathbf{y})|$,
let $\str(\mathbf{x},\mathbf{y})$ denote the string obtained by
concatenating the labels of the edges from $\mathbf{x}$ to $\mathbf{y}$.
For any node $\mathbf{x}$ of $\inputtrie$,
let $\str(\mathbf{x}) = \str(\mathbf{x}, \mathbf{r})$.
%We say that node $\mathbf{x}$ \emph{reprsents} $\str(\mathbf{x})$,
%and we will sometimes identify $\mathbf{x}$ with $\str(\mathbf{x})$
%when no confusions occur.

Let $N$ be the number of nodes in $\inputtrie$.
We associate a unique \emph{id} to each node of $\inputtrie$.
Here we use a bottom-up level-order traversal rank as
the id of each node in $\inputtrie$,
and we sometimes identify each node with its id.
For each node id $i$~($1 \leq i \leq N$) let $\inputtrie[i..] = \str(i)$,
i.e., $\inputtrie[i..]$ is the path string from node $i$
to the root $\mathbf{r}$.

%An example of a trie $\mathcal{T}$ is illustrated in Fig.~\ref{fig:cst}.

%\begin{figure}[t]
%  \centerline{
%    \includegraphics[width=0.6\textwidth]{cst.eps}
%  }
%  \caption{
%    The CS trie for a set 
%    $\{\mathtt{xaxxx}, \mathtt{yaxx}, \mathtt{zaxx}, \mathtt{zyx}, \mathtt{yyy}, 
%    \mathtt{yayy}, \mathtt{xayy}, \mathtt{xzy}, \mathtt{yayxz}, \mathtt{xaxz}\}$
%    of 10 strings over the alphabet $\Sigma = \{\mathtt{a}, \mathtt{x}, \mathtt{y}, \mathtt{z}\}$.
%  }
%  \label{fig:cst}
%\end{figure}

\subsection{Cartesian-tree Pattern Matching}

The \emph{Cartesian tree} of a string $S$,
denoted $\CT(S)$, is the rooted tree with $|S|$ nodes
which is recursively defined as follows:
\begin{itemize}
  \item If $|S| = 0$, then $\CT(S)$ is the empty tree.
  \item If $|S| \geq 1$, then
$\CT(S)$ is the tree whose root $r$ stores the left-most minimum value $S[i]$ in $S$, namely, $r = S[i]$ iff $S[i] \leq S[j]$ for any $i \neq j$ and $S[h] > S[i]$ for any $h < i$. The left-child of $r$ is $\CT(S[1..i-1])$ and the right-child of $r$ is $\CT(S[i+1..|S|])$.
\end{itemize}

The \emph{parent distance encoding} of a string $S$ of length $n$,
denoted $\PD(S)$, is a sequence of $n$ integers over $[0..n-1]$ such that
\[
\PD(S)[i] =
\begin{cases}
  i - \max_{1 \leq j < i} \{j \mid S[j] \leq S[i]\} & \mbox{if such $j$ exists}, \\
  0 & \mbox{otherwise.}
\end{cases}
\]
Namely, $\PD(S)[i]$ represents the distance to from position $i$
to its nearest left-neighbor position $j$ that stores a value that is less than or equal to $S[i]$.

% It is known that $\PD(S)[i]$ is equal to the distance between
% the node $S[i]$ of $\CT(S[1..i])$ and its parent
% (note that in this case, $S[i]$ is always the right child of its parent),
% thus its name ``parent distance''.

A tight connection between $\CT$ and $\PD$ is known:

\begin{lemma}[\cite{SongGRFLP21}] \label{lem:CT-PD}
  For any two strings $S_1$ and $S_2$ of equal length,
  $\CT(S_1) = \CT(S_2)$ iff $\PD(S_1) = \PD(S_2)$.
\end{lemma}

For two strings $S_1$ and $S_2$,
we write $S_1 \approx S_2$ iff $\CT(S_1) = \CT(S_2)$
(or equivalently $\PD(S_1) = \PD(S_2)$).
We also say that $S_1$ and $S_2$ \emph{ct-match} when $S_1 \approx S_2$.
See Fig.~\ref{fig:CT_PD} for a concrete example.

\begin{figure}[tb]
  \centerline{
    \includegraphics[scale=0.4]{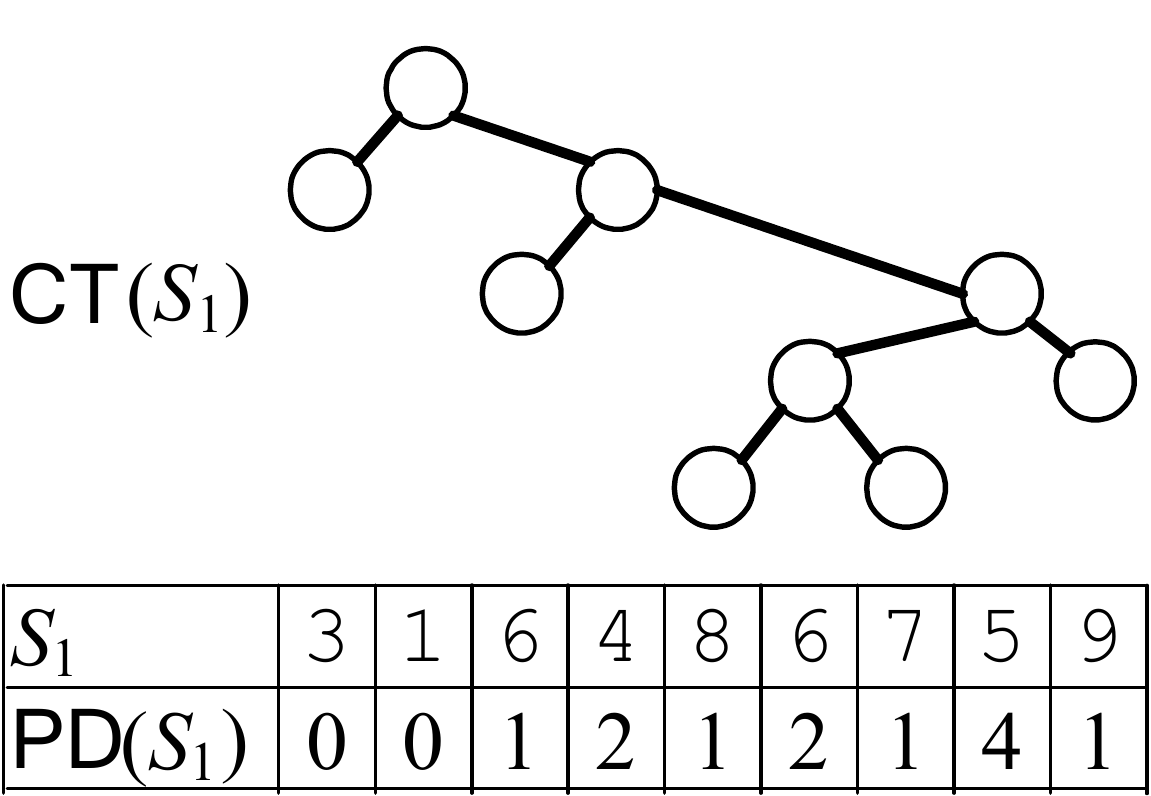}
    \hfil
    \includegraphics[scale=0.4]{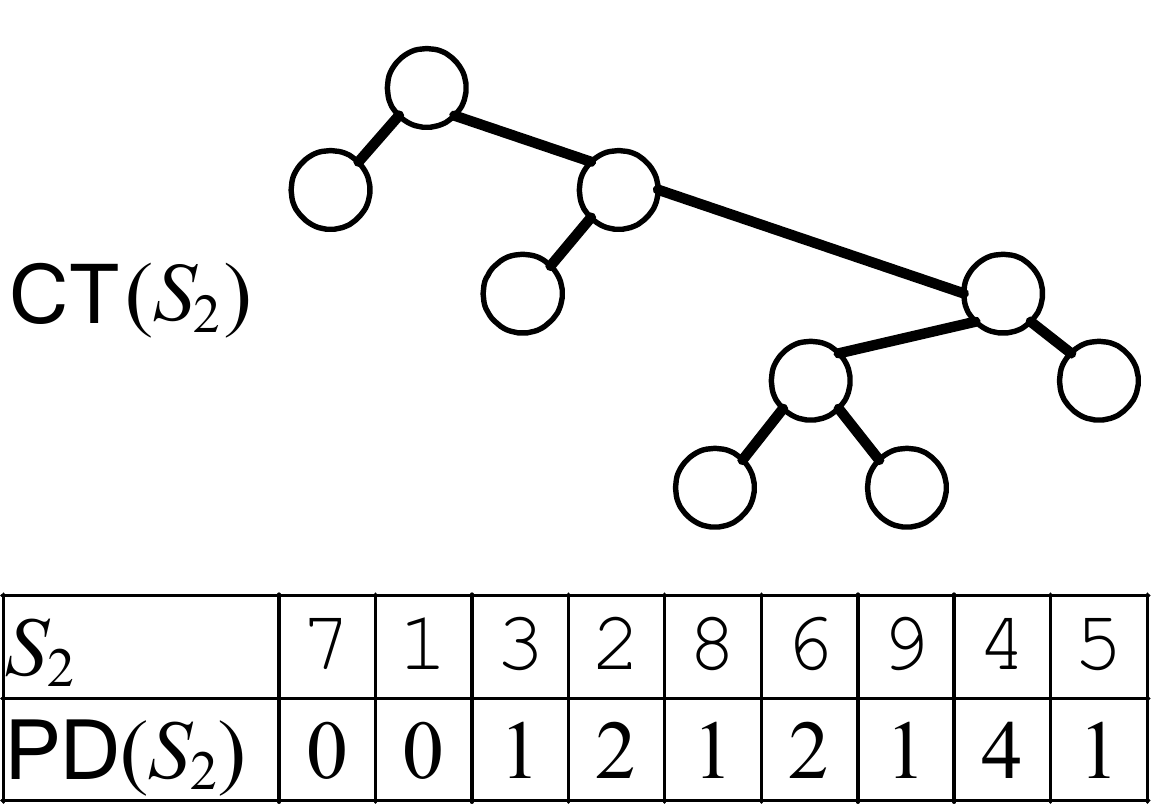}
  }
  \caption{Two strings $S_1 = \mathtt{316486759}$ and $S_2 = \mathtt{713286945}$
  ct-match since $\CT(S_1) = \CT(S_2)$ and $\PD(S_1) = \PD(S_2)$.
  }
  \label{fig:CT_PD}
\end{figure}

We consider the indexing problems for Cartesian-tree pattern matching on a text string and a text trie, which are respectively defined as follows:
\begin{problem}[Cartesian-Tree Pattern Matching on Text String] \label{prob:string} % \\
%\noindent \textbf{Preprocess:} A text string $S$ of length $n$. \\
%\noindent \textbf{Query:} A pattern string $P$ of length $m$. \\
%\noindent \textbf{Report:} All text positions $i$ such that $S[i..i+m-1] \approx P$.
  \begin{description}
  \item[Preprocess:] A text string $S$ of length $n$.
  \item[Query:] A pattern string $P$ of length $m$.
  \item[Report:] All text positions $i$ such that $S[i..i+m-1] \approx P$.
  \end{description}
\end{problem}

\begin{problem}[Cartesian-Tree Pattern Matching on Text Trie] \label{prob:trie} % \\
%\noindent \textbf{Preprocess:} A text trie $\inputtrie$ with $N$ nodes. \\
%\noindent \textbf{Query:} A pattern string $P$ of length $m$. \\
%\noindent \textbf{Report:} All trie nodes $i$ such that $(\inputtrie[i..])[1..m] \approx P$.
  \begin{description}
  \item[Preprocess:] A text trie $\inputtrie$ with $N$ nodes.
  \item[Query:] A pattern string $P$ of length $m$.
  \item[Report:] All trie nodes $i$ such that $(\inputtrie[i..])[1..m] \approx P$.
  \end{description}
\end{problem}

\subsection{Sequence Hash Trees}

Let $\mathcal{W} = \langle w_1, \ldots, w_k \rangle$ be a sequence of
non-empty strings such that for any $1 < i \leq k$, 
$w_i \notin \Prefix(w_j)$ for any $1 \leq j < i$.
The \emph{sequence hash tree}~\cite{coffman} of a sequence
$\mathcal{W} = \langle w_1, \ldots, w_k \rangle$ of $k$ strings,
denoted $\SHT(\mathcal{W}) = \SHT(\mathcal{W})^{k}$,
is a trie structure that is incrementally built as follows:
%is a trie structure that is recursively defined as follows:
%Let $\SHT(\mathcal{W})^{i} = (V_i, E_i)$.
%Then 
%\[
% \SHT(\mathcal{W})^{i} = 
%  \begin{cases}
%   (\{\varepsilon\}, \emptyset) & \mbox{if $i = 0$}, \\
%   (V_{i-1} \cup \{u_{i}\}, E_{i-1} \cup \{(p_{i}, a, u_{i})\}) & \mbox{if $1 \l%eq i \leq k$},
%  \end{cases}
%\]
%where $p_{i}$ is the longest prefix of $w_i$ which satisfies $p_{i} \in V_{i-1}$%,
%$a = w_i[|p_{i}|+1]$,
%and $u_{i}$ is the shortest prefix of $w_i$ which satisfies $u_{i} \notin V_{i-1%}$.

%Intuitively,
%$\SHT(\mathcal{W})$ for the given sequence
%$\mathcal{W} = \langle w_1, \ldots, w_k \rangle$ of strings
%is a rooted tree that is constructed by the following algorithm:
\begin{enumerate}
  \item $\SHT(\mathcal{W})^0 = \SHT(\langle \ \rangle)$ for the empty sequence $\langle \ \rangle$ is the tree only with the root.
  \item For $i = 1, \ldots, k$, $\SHT(\mathcal{W})^i$ is obtained by inserting  the shortest prefix $u_{i}$ of $w_i$ that does not exist in $\SHT(\mathcal{W})^{i-1}$. This is done by finding the longest prefix $p_i$ of $w_i$ that exists in $\SHT(\mathcal{W})^{i-1}$, and adding the new edge $(p_i, c, u_i)$, where $c = w_i[|p_{i}|+1]$ is the first character of $w_i$ that could not be traversed in $\SHT(\mathcal{W})^{i-1}$.
\end{enumerate}

Since we have assumed that each $w_i$ in $\mathcal{W}$ is not a prefix of 
$w_j$ for any $1 \leq j < i$, 
the new edge $(p_{i}, c, u_{i})$ is always created for each $1 \leq i \leq k$.
This means that $\SHT(\mathcal{W})$ contains
exactly $k+1$ nodes (including the root).

To perform pattern matching queries efficiently,
each node of $\SHT(\mathcal{W})$ is augmented with
the \emph{maximal reach pointer}.
For each $1 \leq i \leq k$,
let $u_i$ be the newest node in $\SHT(\mathcal{W})^{i}$,
namely, $u_i$ is the shortest prefix of $w_i$
which did not exist in $\SHT(\mathcal{W})^{i-1}$.
Then, in the complete sequence hash tree
$\SHT(\mathcal{W}) = \SHT(\mathcal{W})^k$,
we set $\mrp(u_i) = u_j$ iff $u_j$ is the deepest node
in $\SHT(\mathcal{W})$ such that $u_j$ is a prefix of $w_i$.
Intuitively, $\mrp(u_i)$ represents
the last visited node $u_j$ when we traverse $w_i$
from the root of the complete $\SHT(\mathcal{W})$.
Note that $j \geq i$ always holds.
When $j = i$ (i.e. when the maximal reach pointer is a self-loop),
then we can omit it because it is not used in the pattern matching algorithm.
%
%See Fig.~\ref{fig:sht} for a concrete example
%of an SHT augmented with the maximal reach pointers.

\section{Cartesian-tree Position Heaps for Strings}
\label{sec:cph_string}

\begin{wrapfigure}[21]{r}{0.7\textwidth}
  \centering
%    \vspace*{-8.5mm}
    \raisebox{5mm}{\includegraphics[scale=0.4]{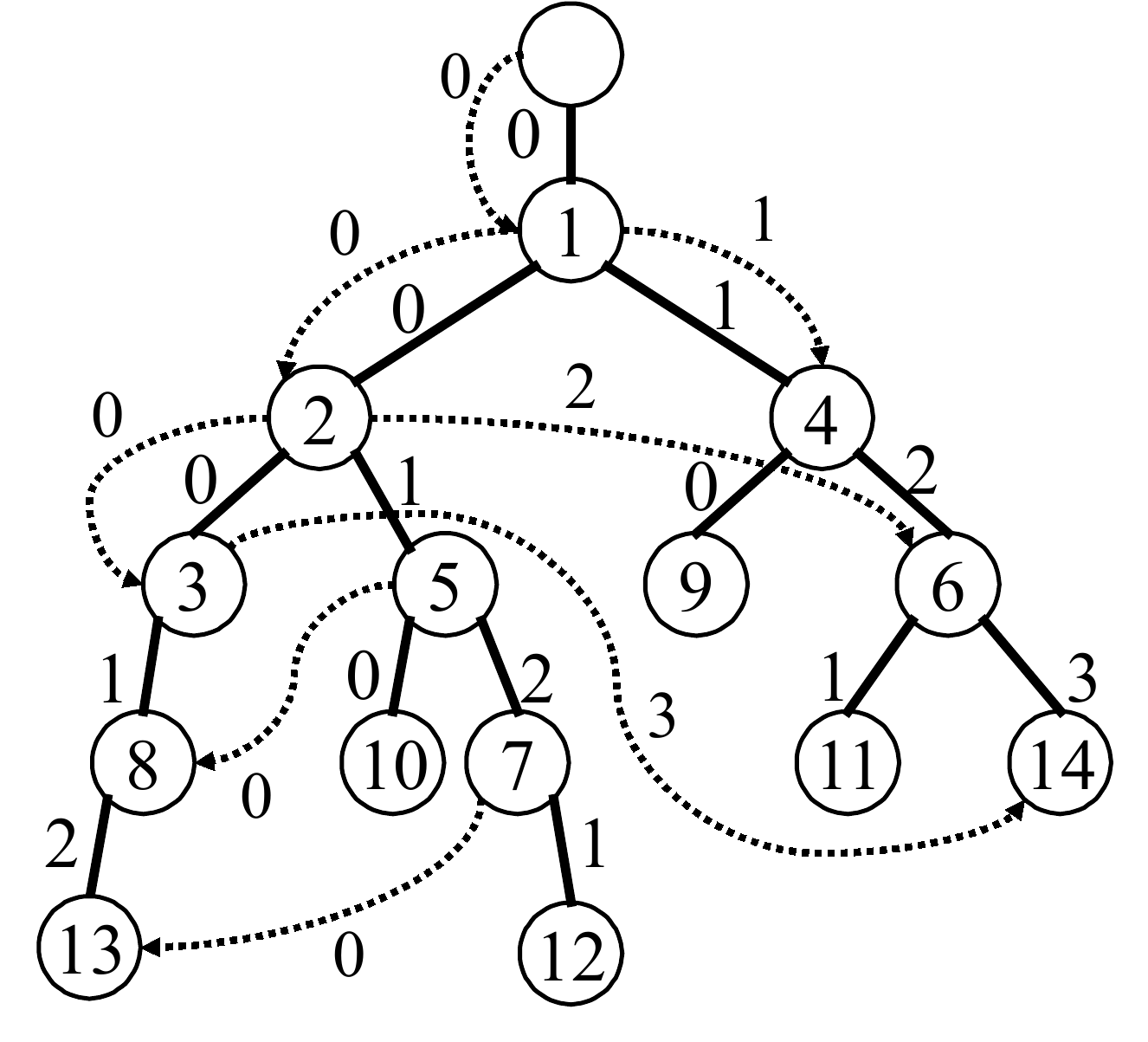}}
    \hfill
    \includegraphics[scale=0.5]{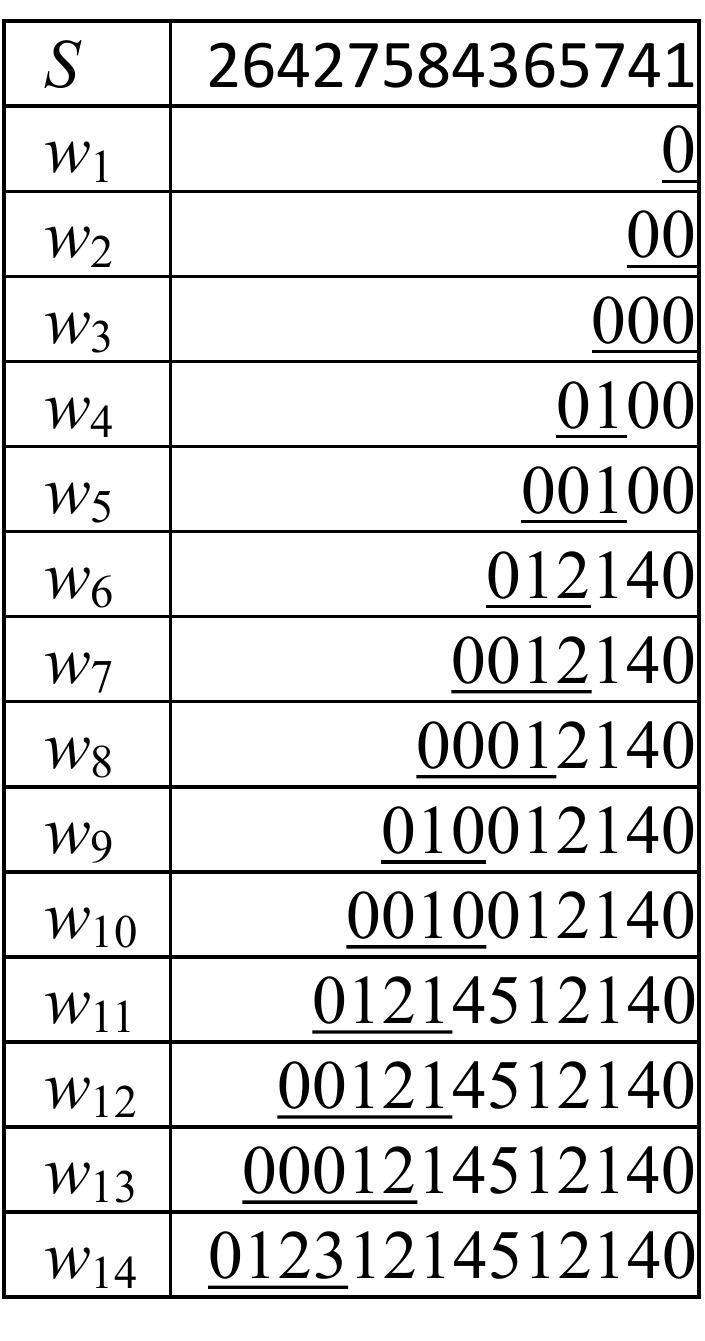}
  \caption{$\CPH(S)$ for string $S = \mathtt{26427584365741}$. For each $w_i = \PD(S[n-i+1..])$, the underlined prefix is the string that is represented by the node $u_i$ in $\CPH(S)$. The dotted arcs are reversed suffix links (not all reversed suffix links are shown).}
  \label{fig:CPH_string_example}
\end{wrapfigure}
In this section, we introduce our new indexing structure for Problem~\ref{prob:string}.
For a given text string $S$ of length $n$,
let $\mathcal{W}_S$ denote the sequence of
the parent distance encodings of the non-empty suffixes of $S$
which are sorted in increasing order of their lengths.
Namely,
$\mathcal{W}_S = \langle w_1$, \ldots, $w_n \rangle = \langle \PD(S[n..])$, \ldots, $\PD(S[1..]) \rangle$,
where $w_{n-i+1} = \PD(S[i..])$.
The \emph{Cartesian-tree Position Heap} (\emph{CPH}) of string $S$,
denoted $\CPH(S)$, is the sequence hash tree of $\mathcal{W}_S$,
that is, $\CPH(S) = \SHT(\mathcal{W}_S)$.
Note that for each $1 \leq i \leq n+1$,
$\CPH(S[i..]) = \SHT(\mathcal{W}_S)^{n-i+1}$ holds.

Our algorithm builds $\CPH(S[i..])$ for decreasing $i = n, \ldots, 1$,
which means that we process the given text string $S$ in a
right-to-left online manner,
by prepending the new character $S[i]$ to the current suffix $S[i+1..]$.

For a sequence $v$ of integers,
let $\Zero_v$ denote the sorted list of positions $z$ in $v$
such that $v[z] = 0$ iff $z \in \Zero_v$.
Clearly $|\Zero_v|$ is equal to the number of $0$'s in $v$.

\begin{lemma} \label{lem:number_of_zeros}
  For any string $S$, $|\Zero_{\PD(S)}| \leq \sigma_S$. 
\end{lemma}

\begin{proof}
Let $\Zero_{\PD(S)} = z_1, \ldots, z_{\ell}$.
We have that $S[z_{1}] > \cdots > S[z_{\ell}]$
since otherwise $\PD(S)[z_{x}] \neq 0$ for some $z_x$, a contradiction.
Thus $|\Zero_{\PD(S)}| \leq \sigma_S$ holds.
\end{proof}

%The next lemma is important for our algorithm.

\begin{lemma} \label{lem:online_PD}
  For each $i = n, \ldots, 1$,
  $\PD(S[i..])$ can be computed from $\PD(S[i+1..])$
  in an online manner, using a total of $O(n)$ time with $O(\sigma_S)$ working space.
\end{lemma}

\begin{proof}
  Given a new character $S[i]$,
  we check each position $z$ in the list $\Zero_{\PD(S[i+1..])}$ in increasing order.
  Let $\hat{z} = z+i$, i.e.,
  $\hat{z}$ is the global position in $S$ corresponding to $z$ in $S[i+1..]$.
  If $S[i] \leq S[\hat{z}]$,
  then we set $\PD(S[i..])[z-i+1] = z-i~(> 0)$
  and remove $z$ from the list.
  Remark that these removed positions
  correspond to the front pointers in the next suffix $S[i..]$.
  We stop when we encounter the first $z$ in the list such that $S[i] > S[\hat{z}]$.
  Finally we add the position $i$ to the head of the remaining positions in
  the list. This gives us $\Zero_{\PD(S[i..])}$
  for the next suffix $S[i..]$.

%  It is clear that for any $j'$ such that $\PD(S[i+1..])[j'] = y > 0$,
%  $\PD(S[i..])[j'] = y$.
%  Moreover, we have $\PD(S[i'..])[j'] = y$ for any $i' < i$ as well,
%  i.e., any non-zero value in the PD encoding never changes.
  It is clear that once a position in the PD encoding
  is assigned a non-zero value, then the value never changes
  whatever characters we prepend to the string.
  Therefore, we can compute $\PD(S[i..])$ from $\PD(S[i+1..])$
  in a total of $O(n)$ time for every $1 \leq i \leq n$.
  The working space is $O(\sigma_S)$ due to Lemma~\ref{lem:number_of_zeros}.
\end{proof}

A position $i$ in a sequence $u$ of non-negative
integers is said to be a \emph{front pointer}
in $u$ if $i - u[i] = 1$
\sinote*{added}{and $i \geq 2$.}
Let $\Front_u$ denote the sorted list of front pointers in $u$.
\annote*{added}{%
	For example, if $u = 01214501$, then $\Front_u = \{2,3,5,6\}$. 
}%
The positions of the suffix $S[i+1..]$ which are removed
from $\Zero_{\PD(S[i+1..])}$ correspond to the front pointers
in $\Front_{\PD(S[i..])}$ for the next suffix $S[i..]$.

Our construction algorithm updates 
$\CPH(S[i+1..])$ to $\CPH(S[i..])$ by inserting a new node for the
next suffix $S[i..]$,
processing the given string $S$ in a right-to-left online manner.
Here the task is to efficiently locate the parent of the new node in the current CPH at each iteration.

As in the previous work on right-to-left online construction of
indexing structures for other types of pattern matching~\cite{Weiner,ehrenfeucht_position_heaps_2011,FujisatoNIBT18,FujisatoNIBT19b},
we use the \emph{reversed suffix links} in our construction algorithm for $\CPH(S)$.
For ease of explanation, we first introduce the notion of the \emph{suffix links}.
Let $u$ be any non-root node of $\CPH(S)$.
We identify $u$ with the path label from the root of $\CPH(S)$ to $u$,
so that $u$ is a PD encoding of some substring of $S$.
We define the suffix link of $u$, denoted $\slink(u)$,
such that $\slink(u) = v$ iff $v$ is obtained by
(1) removing the first $0$~($= u[1]$), and
(2) substituting $0$ for the character $u[f]$ at every front pointer $f \in \Front_u \subseteq [2..|u|]$ of $u$.
The reversed suffix link of $v$ with non-negative integer label $a$,
denoted $\rslink(v, a)$, is defined such that
$\rslink(v, a) = u$ iff $\slink(u) = v$ and $a = |\Front_u|$.
See also Figure~\ref{fig:CPH_string_example}.
%The next lemma is immediate:

\begin{lemma} \label{lem:upper_bound_a}
  Let $u, v$ be any nodes of $\CPH(S)$ such that 
  $\rslink(v, a) = u$ with label $a$.
  Then $a \leq \sigma_S$.
\end{lemma}

\begin{proof}
%  We have $|\Front_u| \leq |\Zero_v|$,
%  where the difference $1$ is due to the above operation (1)
%  of removing the first $0 = u[1]$ from $u$.
  Since $|\Front_u| \leq |\Zero_v|$,
  using Lemma~\ref{lem:number_of_zeros},
  we obtain $a = |\Front_u| \leq |\Zero_v| \leq \sigma_{S'} \leq \sigma_S$,
  where $S'$ is a substring of $S$ such that $\PD(S') = v$.
  \qed
\end{proof}

%Note that the \emph{universe} for the labels
%of the reversed suffix links of $\CPH(S)$ is $[0..n-1]$
%(the maximum value is $n-1$ since $|v| \leq n-1$ for any $\rslink(v, a)$).
%However, the next lemma shows that the number of
The next lemma shows that the number of
out-going reversed suffix links of each node $v$ is bounded by the alphabet size.

%Suppose we have already $\CPH(S)^{i+1}$.

%For each node $v$ of the CPH,
%its reversed suffix link with an integer label $a$~($\geq 0$)
%is denoted $\rslink(a, v)$, and it points to some node of depth $|v|+1$.
%If $\rslink(a, v) = u$,
%then it means that by changing the first $a$ $0$'s in $v$
%to the appropriate non-zero values we obtain $u[2..|u|]$.
%If such a node $u$ does not exist, then we leave $\rslink(a, v)$ undefined.
%We remark that in the special case where $a = 0$,
%we have $\rslink(0, v) = 0v$.

%\begin{wrapfigure}[16]{r}{0.42\textwidth}
\begin{figure}[tbh]
  \centering
%    \vspace*{-9mm}
    \includegraphics[scale=0.4]{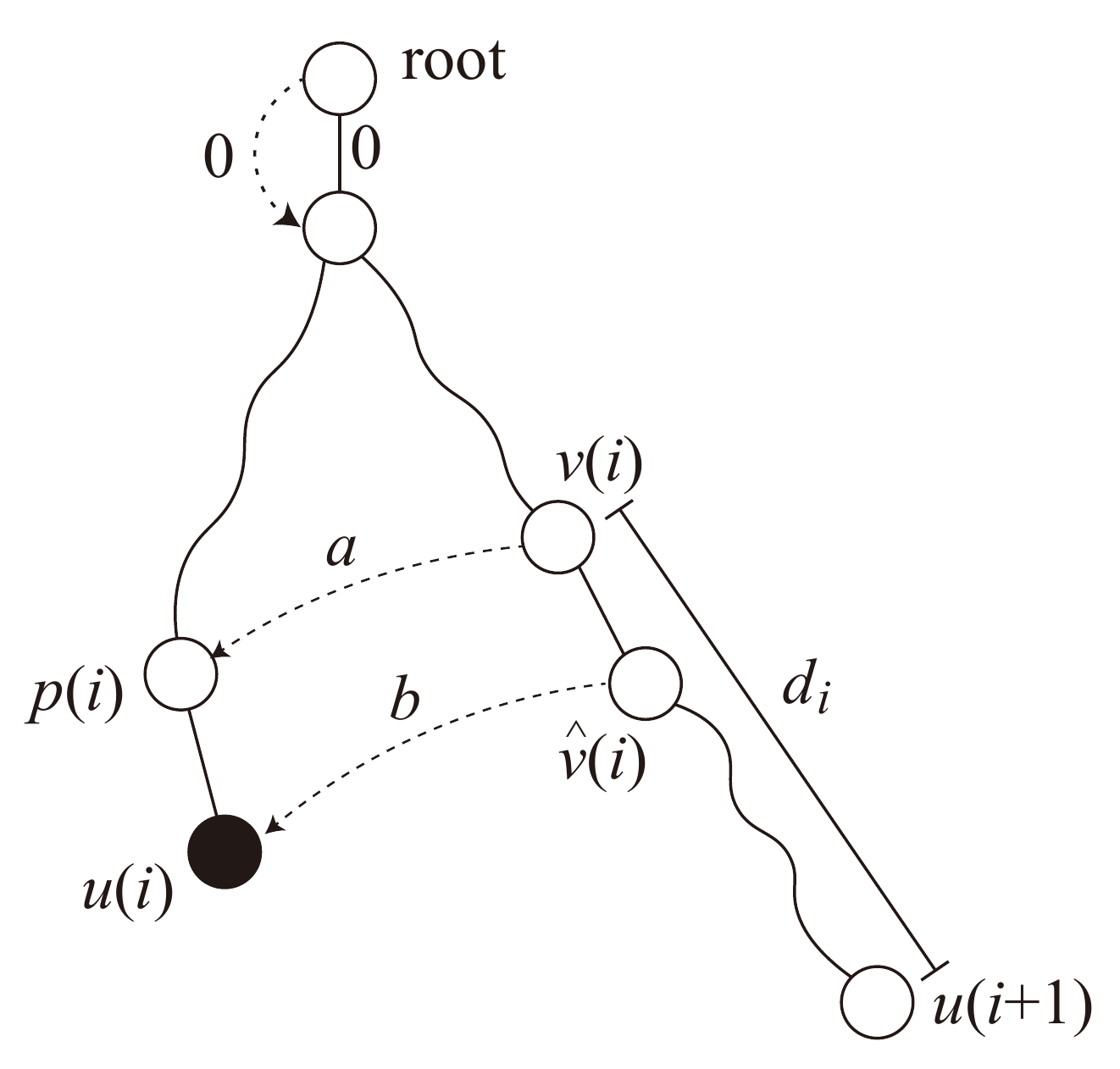}
  \caption{We climb up the path from $u(i+1)$ and find the parent $p(i)$ of the new node $u(i)$ (in black). The label $a$ of the reversed suffix link we traverse from $v(i)$ is equal to the number of front pointers in $p(i)$.}
  \label{fig:CPH_construction_string}
\end{figure}
%\end{wrapfigure}

Our CPH construction algorithm makes use of the following monotonicity of the labels of reversed suffix links:
\begin{lemma} \label{lem:monotonicity}
  Suppose that there exist two reversed suffix links
  $\rslink(v, a) = u$ and $\rslink(v', a') = u'$
  such that $v' = \parent(v)$ and $u' = \parent(u)$.
  Then, $0 \leq a - a' \leq 1$.
\end{lemma}

\begin{proof}
  Immediately follows from
  $a = |\Front_{u}|$, $a' = |\Front_{u'}|$,
  and $u' = u[1..|u|-1]$.
\end{proof}

We are ready to design our right-to-left online construction
algorithm for the CPH of a given string $S$.
Since $\PD(S[i..])$ is the $(n-i+1)$-th string $w_{n-i+1}$ of
the input sequence $\mathcal{W}_S$,
for ease of explanation, we will use the convention that
$u(i) = u_{n-i+1}$ and $p(i) = p_{n-i-1}$,
where the new node $u(i)$ for $w_{n-i+1} = \PD(S[i..])$
is inserted as a child of $p(i)$.
See Figure~\ref{fig:CPH_construction_string}.

%\begin{itembox}[l]
\vspace*{2mm}
\begin{breakbox}
  \noindent Algorithm 1: Right-to-Left Online Construction of $\CPH(S)$
  \begin{description}
  \item[$i = n$ (base case):] We begin with $\CPH(S[n..])$
  which consists of the root $r = u(n+1)$ and the node $u(n)$
  for the first (i.e. shortest) suffix $S[n..]$ of $S$.
  Since $w_1 = \PD(S[n..]) = \PD(S[n]) = 0$,
  the edge from $r$ to $u(n)$ is labeled $0$.
  Also, we set the reversed suffix link $\rslink(r, 0) = u(n)$.

  \item[$i = n-1, \ldots, 1$ (iteration):]
  Given $\CPH(S[i+1..])$ which consists 
  of the nodes $u(i+1), \ldots, u(n)$,
  which respectively represent some prefixes of
  the already processed strings $w_{n-i}, \ldots, w_1 = \PD(S[i+1..]), \ldots, \PD(S[n..])$,
  together with their reversed suffix links.
  We find the parent $p(i)$ of the new node $u(i)$ for $\PD(S[i..])$,
  as follows:
  We start from the last-created node $u(i+1)$ for the previous $\PD(S[i+1..])$,
  and climb up the path towards the root $r$.
  Let $d_i \in [1..|u(i+1)|]$ be the smallest integer
  such that the $d_i$-th ancestor $v(i) = \anc(u(i+1), d_i)$ of $u(i+1)$
  has the reversed suffix link $\rslink(v(i), a)$
  with the label $a = |\Front_{\PD(S[i..i+|v(i)|])}|$.
%  with the label $a = |\Zero_{\PD(S[i+1..i+|v(i)|])}|$.
  We traverse the reversed suffix link
  from $v(i)$ and let $p(i) = \rslink(v(i), a)$.
  We then insert the new node $u(i)$ as the new child of $p(i)$,
  with the edge labeled $\PD(S[i..])[i+|u(i)|-1]$.
  Finally, we create a new reversed suffix link $\rslink(\hat{v}(i), b) = u(i)$,
  where $\hat{v}(i) = \anc(u(i+1), d_i-1)$ and $\parent(\hat{v}) = v$.
  We set $b \leftarrow a+1$ if the position $i+|p(i)|$ is a front pointer of $\PD(S[i..])$,
  and $b \leftarrow a$ otherwise.
  \end{description}
%  \vspace*{-4mm}
\end{breakbox}
%\end{itembox}

For computing the label $a = |\Front_{\PD(S[i..i+|v(i)|])}|$ efficiently,
we introduce a new encoding $\FP$ that is defined as follows:
For any string $S$ of length $n$, let
\sinote*{fixed}{%
  $\FP(S)[i] = |\Front_{\PD(S[i..n])}|$.
}%
The $\FP$ encoding preserves the ct-matching equivalence:

\begin{lemma} \label{lem:FP_PD}
  For any two strings $S_1$ and $S_2$,
  $S_1 \approx S_2$ iff $\FP(S_1) = \FP(S_2)$. 
\end{lemma}

\begin{proof}
  For a string $S$,
  consider the DAG $\mathsf{G}(S) = (V, E)$
  such that $V = \{1, \ldots, |S|\}$,
  $E = \{(j, i) \mid j = i - \PD(S)[i] \}$
  (see also Figure~\ref{fig:DAG} in Appendix).
  By Lemma~\ref{lem:CT-PD}, for any strings $S_1$ and $S_2$,
  $\mathsf{G}(S_1) = \mathsf{G}(S_2)$ iff $S_1 \approx S_2$.
  Now, we will show there is a one-to-one correspondence
  between the DAG $\mathsf{G}$ and the $\FP$ encoding.

  ($\Rightarrow$)
  We are given $\mathsf{G}(S)$ for some (unknown) string $S$.
  Since $\FP(S)[i]$ is the in-degree of the node $i$ of $\mathsf{G}(S)$,
  $\FP(S)$ is unique for the given DAG $\mathsf{G}(S)$.

  ($\Leftarrow$)
  Given $\FP(S)$ for some (unknown) string $S$,
  we show an algorithm that builds DAG $\mathsf{G}(S)$.
  We first create nodes $V = \{1, \ldots, |S|\}$ without edges,
  where all nodes in $V$ are initially unmarked.
  For each $i = n, \ldots, 1$ in decreasing order,
  if $\FP(S)[i] > 0$,
  then select the leftmost $\FP(S)[i]$ unmarked nodes in the range $[i-1..n]$,
  and create an edge $(i, i')$ from each selected node $i'$ to $i$.
  We mark all these $\FP(S)[i]$ nodes at the end of this step,
  and proceed to the next node $i-1$.
  The resulting DAG $\mathsf{G}(S)$ is clearly unique for a given $\PD(S)$.    
\end{proof}

\sinote*{added}{%
  For computing the label $a = |\Front_{\PD(S[i..i+|v(i)|])}| = \FP(S[i..i+|v(i)|])[1]$
  of the reversed suffix link in Algorithm 1,
  it is sufficient to maintain the induced graph $\mathsf{G}_{[i..j]}$
  of DAG $\mathsf{G}$
  for a variable-length sliding window $S[i..j]$ with the nodes $\{i, \ldots, j\}$.
  This can easily be maintained in $O(n)$ total time.
}%

\begin{theorem} \label{lem:construction_cph_string}
  Algorithm 1 builds $\CPH(S[i..])$ for decreasing $i = n, \ldots, 1$ in
  a total of $O(n \log \sigma)$ time and $O(n)$ space,
  where $\sigma$ is the alphabet size.
\end{theorem}

\begin{proof}
  \textbf{Correctness:}
  Consider the $(n-i+1)$-th step in which we process $\PD(S[i..])$.
  By Lemma~\ref{lem:monotonicity},
  the $d_i$-th ancestor $v(i) = \anc(u(i+1), d_i)$ of $u(i+1)$
  can be found by simply walking up the path from the start node $u(i+1)$.
  Note that there always exists such ancestor $v(i)$ of $u(i+1)$ since
  the root $r$ has the defined reversed suffix link $\rslink(r, 0) = 0$.
  By the definition of $v(i)$ and its reversed suffix link,
  $\rslink(v(i), a) = p(i)$ is the longest prefix of $\PD(S[i..])$
  that is represented by $\CPH(S[i+1..])$ (see also Figure~\ref{fig:CPH_construction_string}).
  Thus, $p(i)$ is the parent of the new node $u(i)$ for $\PD(S[i..])$.
  The correctness of the new reversed suffix link $\rslink(\hat{v}(i),b) = u(i)$
  follows from the definition.

  \noindent \textbf{Complexity:}
  The time complexity is proportional to 
  the total number $\sum_{i = 1}^{n} d_i$ of nodes that we visit for all $i = n, \ldots, 1$.
  Clearly $|u(i)| - |u(i+1)| = d_i-2$.
  Thus, $\sum_{i = 1}^{n} d_i = \sum_{i = 1}^{n}(|u(i)| - |u(i+1)|+2) = |u(1)| - |u(n)| + 2n \leq 3n = O(n)$.
  Using Lemma~\ref{lem:rsl_bounds} and sliding-window $\FP$,
  we can find the reversed suffix links
  in $O(\log \sigma_S)$ time at each of the $\sum_{i = 1}^{n} d_i$
  visited nodes.
  Thus the total time complexity is $O(n \log \sigma_S)$.
  Since the number of nodes in $\CPH(S)$ is $n+1$
  and the number of reversed suffix links is $n$,
  the total space complexity is $O(n)$.
\end{proof}
  
\begin{lemma} \label{lem:many_branches}
  There exists a string $S$ of length $n$
  over a binary alphabet $\Sigma = \{\mathtt{1, 2}\}$
  such a node in $\CPH(S)$ has $\Omega(\sqrt{n})$ out-going edges.
\end{lemma}

\begin{proof}
  Consider string $S = \mathtt{1}\mathtt{121}\mathtt{1221} \cdots \mathtt{1}\mathtt{2}^k\mathtt{1}$.
  Then, for any $1 \leq \ell \leq k$,
  there exist nodes representing $01^{k-2}\ell$
  (see also Figure~\ref{fig:lower_bound_branches} in Appendix).
  Since $k = \Theta(\sqrt{n})$,
  the parent node $01^{k-2}$ has $\Omega(\sqrt{n})$ out-going edges.
\end{proof}

Due to Lemma~\ref{lem:many_branches},
if we maintain a sorted list of out-going edges for each node
during our online construction of $\CPH(S[i..])$,
it would require $O(n \log n)$ time even for a constant-size alphabet.
Still, after $\CPH(S)$ has been constructed,
we can sort all the edges offline, as follows:

\begin{theorem}
  For any string $S$ over an integer alphabet $\Sigma = [1..\sigma]$
  of size $\sigma = n^{O(1)}$,
  the edge-sorted $\CPH(S)$ together with
  the maximal reach pointers
  can be computed in $O(n \log \sigma_S)$ time and $O(n)$ space.
\end{theorem}

\begin{proof}
  We sort the edges of $\CPH(S)$ as follows:
  Let $i$ be the id of each node $u(i)$.
  Then sort the pairs $(i, x)$ of the ids and the edge labels.
  Since $i \in [0..n-1]$ and $x \in [1..n^{O(1)}]$,
  we can sort these pairs in $O(n)$ time by a radix sort.
  The maximal reach pointers can be computed in $O(n \log \sigma_S)$
  time using the reversed suffix links,
  in a similar way to the position heaps for exact matching~\cite{ehrenfeucht_position_heaps_2011}.
\end{proof}
See Figure~\ref{fig:CPH_mrp} in Appendix for an example
of $\CPH(S)$ with maximal reach pointers.

\section{Cartesian-tree Position Heaps for Tries}

%In this section, we introduce our new indexing structure for Problem~\ref{prob:trie}.
Let $\inputtrie$ be the input text trie with $N$ nodes.
A na\"ive extension of our CPH
to a trie would be to build the CPH for
the sequence $\langle \PD(\inputtrie[N..]), \ldots, \PD(\inputtrie[1..]) \rangle$
of the parent encodings of all the path strings of $\inputtrie$
towards the root $\mathbf{r}$.
However, this does not seem to work
because the parent encodings are not consistent for suffixes.
For instance, consider two strings
$\mathtt{1432}$ and $\mathtt{4432}$.
Their longest common suffix 
$\mathtt{432}$ is represented by a single path in a trie $\inputtrie$.
However, the longest common suffix of
$\PD(\mathtt{1432}) = 0123$ and $\PD(\mathtt{4432}) = 0100$ is $\varepsilon$.
Thus, in the worst case,
we would have to consider all the path strings $\inputtrie[N..]$, \ldots, $\inputtrie[1..]$ in $\inputtrie$ separately,
but the total length of these path strings in $\inputtrie$ is $\Omega(N^2)$.

To overcome this difficulty,
\sinote*{modified}{%
we reuse the $\FP$ encoding from Section~\ref{sec:cph_string}.
Since $\FP(S)[i]$ is determined merely by the suffix $S[i..]$,
the $\FP$ encoding is suffix-consistent.
}%
For an input trie $\inputtrie$,
let the \emph{FP-trie} $\fptrie$ be the reversed trie
storing $\FP(\inputtrie[i..])$ for all the original
path strings $\inputtrie[i..]$ towards the root.
Let $N'$ be the number of nodes in $\fptrie$.
Since $\FP$ is suffix-consistent, $N' \leq N$ always holds.
Namely, $\FP$ is a linear-size representation of
the equivalence relation of the nodes of $\inputtrie$ w.r.t. $\approx$.
Each node $v$ of $\fptrie$ stores
the equivalence class
$\mathcal{C}_v = \{i \mid \fptrie[v..] = \FP(\inputtrie[i..])\}$
of the nodes $i$ in $\inputtrie$ that correspond to $v$.
We set $\min\{\mathcal{C}_v\}$ to be the representative of $\mathcal{C}_v$,
as well as the id of node $v$.
See Figure~\ref{fig:input_trie} in Appendix.

Let $\TrieSigma$ be the set of distinct characters
(i.e. edge labels) in $\inputtrie$ and let
$\triesigma = |\TrieSigma|$.
The FP-trie $\fptrie$ can be computed in $O(N\triesigma)$ time and
working space by a standard traversal on $\inputtrie$,
where we store at most $\triesigma$ front pointers
in each node of the current path in $\inputtrie$ due to Lemma~\ref{lem:online_PD}.
%We use a bottom-up level-order traversal rank as
%the id of each node in $\fptrie$.

Let $i_{N'}, \ldots, i_{1}$ be the node id's of $\fptrie$
which are sorted in decreasing order.
The Cartesian-tree position heap for the input trie $\inputtrie$
is $\CPH(\inputtrie) = \SHT(\mathcal{W}_{\inputtrie})$,
where $\mathcal{W}_{\inputtrie} = \langle \PD(\inputtrie[i_{N'}..], \ldots, \PD(\inputtrie[i_1..]) \rangle$.

%We also compute, for each node $i_k$ of the FP-trie $\fptrie$,
%its nearest ancestor $i_{j}$ such that $\Front_{\inputtrie[i_j..]} < \Front_{\inputtrie[i_j..]}$.

As in the case of string inputs in Section~\ref{sec:cph_string},
we insert the shortest prefix of $\PD(\inputtrie[i_{k}..])$
that does not exist in $\CPH(\inputtrie[i_{k+1}..])$.
To perform this insert operation,
we use the following data structure for a random-access
query on the PD encoding of any path string in $\inputtrie$:

\begin{lemma} \label{lem:nearest_ancestor}
  There is a data structure of size $O(N \triesigma)$
  that can answer the following queries in $O(\triesigma)$ time each.\\
  \textbf{Query input:} The id $i$ of a node in $\inputtrie$ and integer $\ell > 0$.\\
  \textbf{Query output:} The $\ell$th (last) symbol $\PD((\inputtrie[i..])[1..\ell])[\ell]$ in $\PD(\inputtrie[i..])[1..\ell]$.
\end{lemma}

\begin{proof}
Let $\mathbf{x}$ be the node with id $i$,
and $\mathbf{z} = \anc(\mathbf{x}, \ell)$.
Namely, $\str(\mathbf{x}, \mathbf{z}) = (\inputtrie[j..])[1..\ell]$.
For each character $a \in \TrieSigma$,
let $\na(\mathbf{x},a)$ denote the nearest ancestor $\mathbf{y}_a$ of $\mathbf{x}$
such that the edge $(\parent(\mathbf{y}_a), \mathbf{y}_a)$ is labeled $a$.
If such an ancestor does not exist, then we set $\na(\mathbf{x},a)$
to the root $\mathbf{r}$.

Let $\mathbf{z'} = \anc(\mathbf{x}, \ell-1)$,
and $b$ be the label of the edge $(\mathbf{z}, \mathbf{z'})$.
Let $D$ be an empty set.
For each character $a \in \TrieSigma$,
we query $\na(\mathbf{x}, a) = \mathbf{y}_a$.
If $d_a = |\mathbf{y}_a| - |\mathbf{z}'| > 0$ and $a \leq b$,
then $d_a$ is a candidate for $(\PD(\inputtrie[j..])[1..\ell])[\ell]$
and add $d_a$ to set $D$.
After testing all $a \in \TrieSigma$,
we have that $(\PD(\inputtrie[j..])[1..\ell])[\ell] = \min D$.
See Figure~\ref{fig:nearest_ancestor}.

%\begin{wrapfigure}[9]{r}{0.5\textwidth}
\begin{figure}[tbh]
  \centering
%    \vspace*{-9mm}
    \includegraphics[scale=0.5]{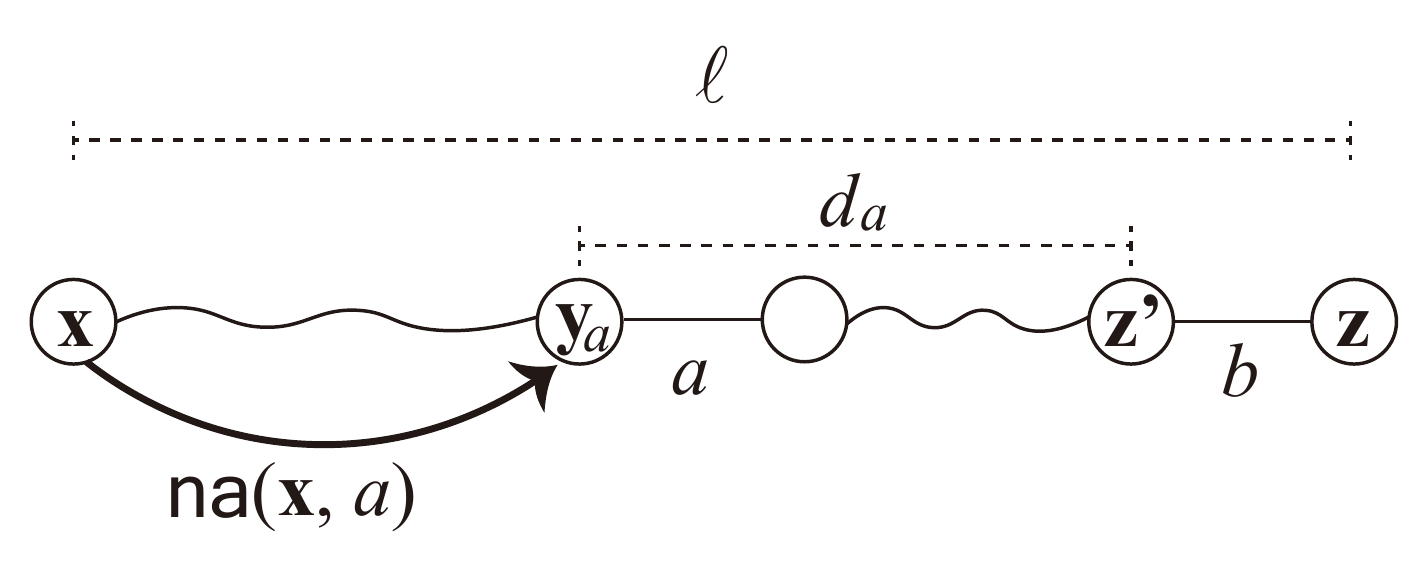}
%    \vspace*{-2mm}
  \caption{Illustration for the data structure of Lemma~\ref{lem:nearest_ancestor}, where $(\inputtrie[i..])[1..\ell] = \str(\mathbf{x}, \mathbf{z})$.}
  \label{fig:nearest_ancestor}
\end{figure}
%\end{wrapfigure}

For all characters $a \in \TrieSigma$ and all nodes $x$ in $\inputtrie$,
$\na(\mathbf{x}, a)$ can be pre-computed in a total of
$O(N \triesigma)$ preprocessing time and space,
by standard traversals on $\inputtrie$.
Clearly each query is answered in $O(\sigma_{\inputtrie})$ time.
\end{proof}
%Recall that we cannot afford to store $\PD(\inputtrie[i..])$ explicitly,
%as it can require $\Omega(N^2)$ space for all $1 \leq i \leq N$.

%The main theorem of this seciton follows:
\begin{theorem}
  Let $\inputtrie$ be a given trie with $N$ nodes
  whose edge labels are from an integer alphabet of size $n^{O(1)}$.
  The edge-sorted $\CPH(\inputtrie)$ with the maximal reach pointers,
  which occupies $O(N \triesigma)$ space,
  can be built in $O(N \triesigma)$ time.
\end{theorem}

\begin{proof}
  The rest of the construction algorithm of $\CPH(\inputtrie)$
  is almost the same as the case of the CPH for a string,
  except that the amortization argument in the proof
  for Theorem~\ref{lem:construction_cph_string} cannot be applied to
  the case where the input is a trie.
  Instead, we use the nearest marked ancestor (NMA)
  data structure~\cite{westbrook_fast_incre_1992,amir_improved_dynamic_1995}
  that supports queries and marking nodes in amortized $O(1)$ time each,
  using space linear in the input tree.
  For each $a \in [0..\triesigma]$,
  we create a copy $\CPH_a(\inputtrie)$ of $\CPH(\inputtrie)$
  and maintain the NMA data structure on $\CPH_a(\inputtrie)$
  so that every node $v$ that has defined reversed suffix link $\rslink(v, a)$
  is marked, and any other nodes are unmarked.
  The NMA query for a given node $v$ with character $a$
  is denoted by $\nma_a(v)$.
  If $v$ itself is marked with $a$, then let $\nma_a(v) = v$.
%  If such an ancestor of $v$ does not exist, let $\nma_a(v) = v$.
  %
  \sinote*{added}{%
  For any node $\mathbf{x}$ of $\inputtrie$,
  let $\mathcal{I}_{\mathbf{x}}$ be the array of size at most $\triesigma$
  s.t. $\mathcal{I}_{\mathbf{x}}[j] = h$
  iff $h$ is the $j$th smallest element of $\Front_{\PD(\str(\mathbf{x}))}$.
%  To access $\mathcal{I}_{\mathbf{x}}[d]$ for a given
%  node $\mathbf{x}$ in $\inputtrie$ and $d > 0$ in $O(1)$ time,
%  we use the dynamic level ancestor data structure
%  on our CPH that allows for leaf insertions
%  and level ancestor queries in $O(1)$ time each~\cite{AlstrupH00}.
  }%
  
  We are ready to design our construction algorithm:
  Suppose that we have already built $\CPH(\inputtrie[i_{k+1}..])$
  and we are to update it to $\CPH(\inputtrie[i_{k}..])$.
%  We begin with the node $u(i_{k+1})$ that corresponds to $\PD(\inputtrie[i_{k+1}..])$.
  \sinote*{modified}{%
  Let $\mathbf{w}$ be the node in $\inputtrie$ with id $i_{k}$,
  and let $\mathbf{u} = \parent(\mathbf{w})$ in $\fptrie$.
  Let $u$ be the node of $\CPH(\inputtrie[i_{k+1}..])$
  that corresponds to $\mathbf{u}$.
  We initially set $v \leftarrow u$ and
  $a \leftarrow |\Front_{\PD(\inputtrie[i_k..i_k+|u|])}|$.
  Let $d(a) = \max\{|u|-\mathcal{I}_{\mathbf{w}}[a]+1, 0\}$.
  %  and $a \leftarrow \sigma_{\inputtrie}$ (the latter is due to Lemma~\ref{lem:upper_bound_a}).
  %  We climb up the path from $u$ as follows:
  We perform the following:
  \begin{enumerate}
  \item[(1)] Check whether $v' = \anc(u, d(a))$ is marked in $\CPH_a(\inputtrie)$. If so, go to (2). Otherwise, update $v \leftarrow v'$, $a \leftarrow a - 1$, and repeat (1).
  \item[(2)] Return $\nma(v, a)$.
%  \item[(1)] Update $v \leftarrow \nma_a(v)$. If $|v| \geq a$, return $v$ and break. Otherwise, go to (2).
%  \item[(2)] Update $v \leftarrow \anc(v, |v|-a)$. Update $a \leftarrow a - 1$, and go to (1).
  \end{enumerate}
  By the definitions of $\mathcal{I}_{\mathbf{w}}[a]$ and $d(a)$,
  the node $v(i_k)$ from which we should take the reversed suffix link
  is in the path between $v'$ and $v$,
  and it is the lowest ancestor of $v$ that has
  the reversed suffix link with $a$.
  Thus, the above algorithm correctly computes the desired node.
%  The node $v$ returned by the above algorithm
%  is exactly the node $v(i_k)$ from which we take the reversed suffix link.
%  The level ancestor query of (2) skips
%  redundant nodes in the climbing path from $u(i_{k+1})$.
  By Lemma~\ref{lem:upper_bound_a},
  the number of queries in (1) for each of the $N'$ nodes is $O(\triesigma)$,
  and we use the dynamic level ancestor data structure
  on our CPH that allows for leaf insertions
  and level ancestor queries in $O(1)$ time each~\cite{AlstrupH00}.
  This gives us $O(N \triesigma)$-time and space construction.
  }%

  We will reuse the random access data structure of Lemma~\ref{lem:nearest_ancestor}
  for pattern matching (see Section~\ref{sec:matching_trie}).
  Thus $\CPH(\inputtrie)$ requires $O(N \sigma_{\inputtrie})$ space.
\end{proof}

%\begin{problem}
%  \begin{description}
%  \item[Query input:] The id $j$ of a node in $\inputtrie$ and integer $\ell > 0%$.
%  \item[Query output:] The last $\ell$th symbol $\PD(\inputtrie[i..i+\ell-1])[\e%ll]$ in $\PD(\inputtrie[i..\ell])$.
%  \end{description}
%\end{problem}

%For each $a \in \TrieSigma$,
%let $\inputtrie_{a}$ be the induced tree of $\inputtrie$
%that consists of the root and all the nodes $\mathbf{v}$ such that
%the edge $(\parent(\mathbf{v}), \mathbf{v})$ is labeled $a$.
%We create a link $\ell_a(\mathbf{v}) = \mathbf{u}$
%if $\mathbf{u}$ is the neartest ancestor of $\mathbf{v}$
%that exists in $\inputtrie_a$.

%In what follows, we present an $O(N \triesigma)$-space data structure
%that supports the above query in $O(\triesigma)$ time.

\section{Cartesian-tree Pattern Matching with Position Heaps}
\label{sec:matching}

\subsection{Pattern Matching on Text String $S$ with $\CPH(S)$}
\label{sec:matching_string}

Given a pattern $P$ of length $m$, we first compute
the greedy factorization $\fact(P) = P_0, P_1, \ldots, P_k$ of $P$
such that $P_0 = \varepsilon$,
and for $1 \leq l \leq k$,
$P_l = P[\lensum(l-1)+1..\lensum(l)]$ is the longest prefix of
$P_l \cdots P_k$ that is represented by $\CPH(S)$,
where $\lensum(l) = \sum_{j = 0}^l|P_j|$.
We consider such a factorization of $P$
since the height $h$ of $\CPH(S)$ can be smaller than the pattern length $m$.
%The next lemma is a key for our pattern matching algorithm.

\begin{lemma} \label{lem:number_of_branches}
% Let $S$ be a string over $\Sigma$.
 Any node $v$ in $\CPH(S)$ has at most $|v|$ out-going edges.
\end{lemma}
\begin{proof}
  Let $(v, c, u)$ be any out-going edge of $v$.
%  and let $S'$ be a substring of $S$ such that $u = \PD(S')$.
  When $|u|-1$ is a front pointer of $u$,
  then $c = u[|u|]$ and this is when $c$ takes the maximum value.
  Since $u[|u|] \leq |u|-1$, we have $c \leq |u|-1$.
  Since the edge label of $\CPH(S)$ is non-negative,
  $v$ can have at most $|u|-1 = |v|$ out-going edges.
%  By the definition of PD for a string $s$ $0\leq \PD(s)[i]\leq i$ 
\end{proof}

The next corollary immediately follows from Lemma~\ref{lem:number_of_branches}.
\begin{corollary}
  Given a pattern $P$ of length $m$,
  its factorization $\fact(P)$ can be computed in
  $O(m \log (\min\{m, h\}))$ time,
  where $h$ is the height of $\CPH(S)$.
\end{corollary}

The next lemma is analogous to the position heap for exact matching~\cite{ehrenfeucht_position_heaps_2011}.
\begin{lemma} \label{lemma:decendant_condition}
  Consider two nodes $u$ and $v$ in $\CPH(S)$
  such that $u = \PD(P)$ the id of $v$ is $i$.
  Then, $\PD(S[i..])[1..|u|] = u$ iff
  one of the following conditions holds:
  (a) $v$ is a descendant of $u$;
  (b) $\mrp(v)$ is a descendant of $u$.
\end{lemma}
See also Figure~\ref{fig:matching} in Appendix.
We perform a standard traversal on $\CPH(S)$ so that
one we check whether a node is a descendant of
another node in $O(1)$ time.

When $k = 1$~(i.e. $\fact(P) = P$),
$\PD(P)$ is represented by some node $u$ of $\CPH(S)$.
Now a direct application of Lemma~\ref{lemma:decendant_condition}
gives us all the $\occ$ pattern occurrences 
in $O(m \log m + \occ)$ time, where $\min\{m, h\} = m$ in this case.
All we need here is to report the id of every descendant of $u$
(Condition (a)) and the id of each node $v$ that satisfies Condition (b).
The number of such nodes $v$ is less than $m$.

When $k \geq 2$~(i.e. $\fact(P) \neq P$),
there is no node that represents $\PD(P)$ for the whole pattern $P$.
This happens only when $\occ < m$,
since otherwise there has to be a node representing $\PD(P)$
by the incremental construction of $\CPH(S)$,
a contradiction.
This implies that Condition (a) of Lemma~\ref{lemma:decendant_condition}
does apply when $k \geq 2$.
Thus, the \emph{candidates} for the pattern occurrences
only come from Condition (b),
which are restricted to the nodes $v$
such that $\mrp(v) = u_1$, where $u_1 = \PD(P_1)$.
%Recall that $P_1$ is the longest proper prefix of $P$
%whose parent encoding is represented by $\CPH(S)$.
We apply Condition (b) iteratively for
the following $P_2, \ldots, P_k$,
while keeping track of the position $i$ that was associated to each node $v$
such that $\mrp(v) = u_1$.
This can be done by padding $i$ with
the off-set $\lensum(l-1)$ when we process $P_l$.
We keep such a position $i$
if Condition (b) is satisfied for
all the following pattern blocks $P_2, \ldots, P_k$,
namely, if the maximal reach pointer of the
node with id $i+\lensum(l-1)$ points to node $u_l = \PD(P_l)$
for increasing $l = 2, \ldots, k$.
As soon as Condition (b) is not satisfied with some $l$,
we discard position $i$.

Suppose that we have processed the all pattern blocks
$P_1, \ldots, P_k$ in $\fact(P)$.
Now we have that
$\PD(S[i..])[1..m] = \PD(P)$ (or equivalently $S[i..i+m-1] \approx P$)
\emph{only if} the position $i$ has survived.
Namely, position $i$ is only a candidate of a pattern occurrence
at this point,
since the above algorithm only guarantees that
$\PD(P_1) \cdots \PD(P_k) = \PD(S[i..])[1..m]$.
Note also that, by Condition (b),
the number of such survived positions $i$
is bounded by $\min\{|P_1|, \ldots, |P_k|\} \leq m/k$.

For each survived position $i$,
we verify whether $\PD(P) = \PD(S[i..])[1..m]$.
This can be done by checking, for each increasing $l = 1, \ldots, k$,
whether or not $\PD(S[i..])[\lensum(l-1)+y]= \PD(P_1 \cdots P_l)[\lensum(l-1)+y]$
for every position $y$~($1 \leq y \leq |P_l|$) such that $\PD(P_l)[y] = 0$.
%$(\PD(S[i..])[\lensum(l-1)+1..\lensum(l)])[\lensum(l-1)+y]=0$.
By the definition of $\PD$,
the number of such positions $y$ is at most $\sigma_{P_l} \leq \sigma_P$.
Thus, for each survived position $i$
we have at most $k \sigma_P$ positions to verify.
Since we have at most $m/k$ survived positions,
the verification takes a total of $O(\frac{m}{k} \cdot k \sigma_P) = O(m\sigma_P)$ time.

\begin{theorem}
  Let $S$ be the text string of length $n$.
  Using $\CPH(S)$ of size $O(n)$
  augmented with the maximal reach pointers,
  we can find all $\occ$ occurrences for a given pattern $P$
  in $S$ in $O(m(\sigma_P +\log(\min\{m,h\})) + \occ)$ time,
  where $m = |P|$ and $h$ is the height of $\CPH(S)$.
\end{theorem}

\subsection{Pattern Matching on Text Trie $\inputtrie$ with $\CPH(\inputtrie)$}
\label{sec:matching_trie}

In the text trie case,
we can basically use the same matching algorithm as in the text string case
of Section~\ref{sec:matching_string}.
However, recall that we cannot afford to store
the PD encodings of the path strings in $\inputtrie$
as it requires $\Omega(n^2)$ space.
Instead, we reuse the random-access data structure of
Lemma~\ref{lem:nearest_ancestor} for the verification step.
Since it takes $O(\sigma_{\inputtrie})$ time for each random-access query,
and since the data structure occupies $O(N \sigma_{\inputtrie})$ space,
we have the following complexity:
\begin{theorem}
  Let $\inputtrie$ be the text trie with $N$ nodes.  
  Using $\CPH(\inputtrie)$ of size $O(N \sigma_{\inputtrie})$
  augmented with the maximal reach pointers,
  we can find all $\occ$ occurrences for a given pattern $P$
  in $\inputtrie$ in 
  $O(m(\sigma_P \sigma_{\inputtrie}+\log(\min\{m,h\})) + \occ)$ time,
  where $m = |P|$ and $h$ is the height of $\CPH(\inputtrie)$.
\end{theorem}

\section*{Acknowledgments}
This work was supported by JSPS KAKENHI Grant Numbers 
JP18K18002 (YN) and JP21K17705 (YN),
and by JST PRESTO Grant Number JPMJPR1922 (SI).

\clearpage

\bibliographystyle{abbrv}
\bibliography{ref}

\begin{thebibliography}{10}

\bibitem{AlstrupH00}
S.~Alstrup and J.~Holm.
\newblock Improved algorithms for finding level ancestors in dynamic trees.
\newblock In U.~Montanari, J.~D.~P. Rolim, and E.~Welzl, editors, {\em {ICALP}
  2000}, volume 1853 of {\em Lecture Notes in Computer Science}, pages 73--84.
  Springer, 2000.

\bibitem{amir_improved_dynamic_1995}
A.~Amir, M.~Farach, R.~M. Idury, J.~A.~L. Poutr{\'e}, and A.~A. Sch{\"a}ffer.
\newblock Improved dynamic dictionary matching.
\newblock {\em Information and Computation}, 119(2):258--282, 1995.

\bibitem{Baker93}
B.~S. Baker.
\newblock A theory of parameterized pattern matching: algorithms and
  applications.
\newblock In {\em STOC 1993}, pages 71--80, 1993.

\bibitem{baker95parameterized}
B.~S. Baker.
\newblock Parameterized pattern matching by {B}oyer-{M}oore type algorithms.
\newblock In {\em Proc. 6th Annual {ACM}-{SIAM} Symposium on Discrete
  Algorithms}, pages 541--550, 1995.

\bibitem{Baker96}
B.~S. Baker.
\newblock Parameterized pattern matching: Algorithms and applications.
\newblock {\em J. Comput. Syst. Sci.}, 52(1):28--42, 1996.

\bibitem{BenderF04}
M.~A. Bender and M.~Farach-Colton.
\newblock The level ancestor problem simplified.
\newblock {\em Theor. Comput. Sci.}, 321(1):5--12, 2004.

\bibitem{ChoNPS15}
S.~Cho, J.~C. Na, K.~Park, and J.~S. Sim.
\newblock A fast algorithm for order-preserving pattern matching.
\newblock {\em Inf. Process. Lett.}, 115(2):397--402, 2015.

\bibitem{coffman}
E.~Coffman and J.~Eve.
\newblock File structures using hashing functions.
\newblock {\em Communications of the ACM}, 13:427--432, 1970.

\bibitem{CrochemoreIKKLP16}
M.~Crochemore, C.~S. Iliopoulos, T.~Kociumaka, M.~Kubica, A.~Langiu, S.~P.
  Pissis, J.~Radoszewski, W.~Rytter, and T.~Walen.
\newblock Order-preserving indexing.
\newblock {\em Theor. Comput. Sci.}, 638:122--135, 2016.

\bibitem{ehrenfeucht_position_heaps_2011}
A.~Ehrenfeucht, R.~M. McConnell, N.~Osheim, and S.-W. Woo.
\newblock Position heaps: A simple and dynamic text indexing data structure.
\newblock {\em Journal of Discrete Algorithms}, 9(1):100--121, 2011.

\bibitem{FuCLN07}
T.~Fu, K.~F. Chung, R.~W.~P. Luk, and C.~Ng.
\newblock Stock time series pattern matching: Template-based vs. rule-based
  approaches.
\newblock {\em Eng. Appl. Artif. Intell.}, 20(3):347--364, 2007.

\bibitem{FujisatoNIBT18}
N.~Fujisato, Y.~Nakashima, S.~Inenaga, H.~Bannai, and M.~Takeda.
\newblock Right-to-left online construction of parameterized position heaps.
\newblock In {\em PSC 2018}, pages 91--102, 2018.

\bibitem{FujisatoNIBT19b}
N.~Fujisato, Y.~Nakashima, S.~Inenaga, H.~Bannai, and M.~Takeda.
\newblock The parameterized position heap of a trie.
\newblock In {\em {CIAC} 2019}, pages 237--248, 2019.

\bibitem{IIT11}
T.~I, S.~Inenaga, and M.~Takeda.
\newblock Palindrome pattern matching.
\newblock In R.~Giancarlo and G.~Manzini, editors, {\em {CPM} 2011}, volume
  6661 of {\em Lecture Notes in Computer Science}, pages 232--245. Springer,
  2011.

\bibitem{KimH16}
H.~Kim and Y.~Han.
\newblock {OMPPM:} online multiple palindrome pattern matching.
\newblock {\em Bioinform.}, 32(8):1151--1157, 2016.

\bibitem{KimEFHIPPT14}
J.~Kim, P.~Eades, R.~Fleischer, S.~Hong, C.~S. Iliopoulos, K.~Park, S.~J.
  Puglisi, and T.~Tokuyama.
\newblock Order-preserving matching.
\newblock {\em Theor. Comput. Sci.}, 525:68--79, 2014.

\bibitem{MatsuokaAIBT16}
Y.~Matsuoka, T.~Aoki, S.~Inenaga, H.~Bannai, and M.~Takeda.
\newblock Generalized pattern matching and periodicity under substring
  consistent equivalence relations.
\newblock {\em Theor. Comput. Sci.}, 656:225--233, 2016.

\bibitem{ParkBALP20}
S.~G. Park, M.~Bataa, A.~Amir, G.~M. Landau, and K.~Park.
\newblock Finding patterns and periods in cartesian tree matching.
\newblock {\em Theor. Comput. Sci.}, 845:181--197, 2020.

\bibitem{SongGRFLP21}
S.~Song, G.~Gu, C.~Ryu, S.~Faro, T.~Lecroq, and K.~Park.
\newblock Fast algorithms for single and multiple pattern {Cartesian} tree
  matching.
\newblock {\em Theor. Comput. Sci.}, 849:47--63, 2021.

\bibitem{Weiner}
P.~Weiner.
\newblock Linear pattern-matching algorithms.
\newblock In {\em Proc. of 14th IEEE Ann. Symp. on Switching and Automata
  Theory}, pages 1--11, 1973.

\bibitem{westbrook_fast_incre_1992}
J.~Westbrook.
\newblock Fast incremental planarity testing.
\newblock In {\em Proc. ICALP 1992}, number 623 in LNCS, pages 342--353, 1992.

\end{thebibliography}

\clearpage
\appendix

\section{Figures}

% In this appendix we show some supplemental figures.

%\begin{figure}[hp]
%  \centerline{
%    \includegraphics[scale=0.4]{sht.pdf}
%  }
%  \caption{$\SHT(\mathcal{W})$ for $\mathcal{W} = \langle \mathtt{aba, acb, ba, baab, abaa, bca} \rangle$. The black arcs are the edges and the white arcs are the maximal reach pointers. The underlined prefixes of $w_1$ and $w_3$ show the distinations of the maximal reach pointers for $u_1$ and $u_3$, respectively.}
%  \label{fig:sht}
%\end{figure}

\begin{figure}[!h]
  \centerline{
    \includegraphics[scale=0.4]{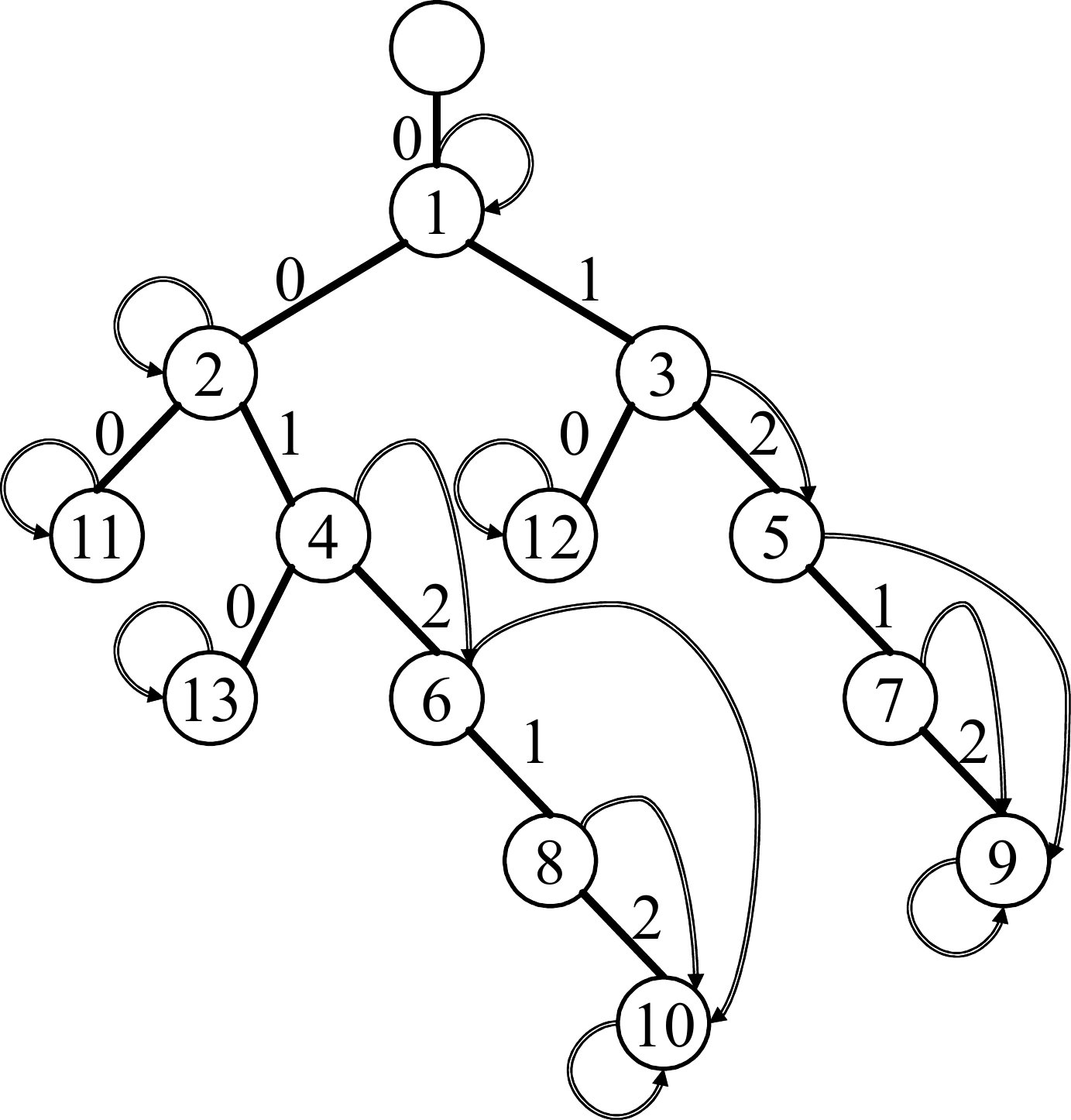}
    \hfill
    \includegraphics[scale=0.38]{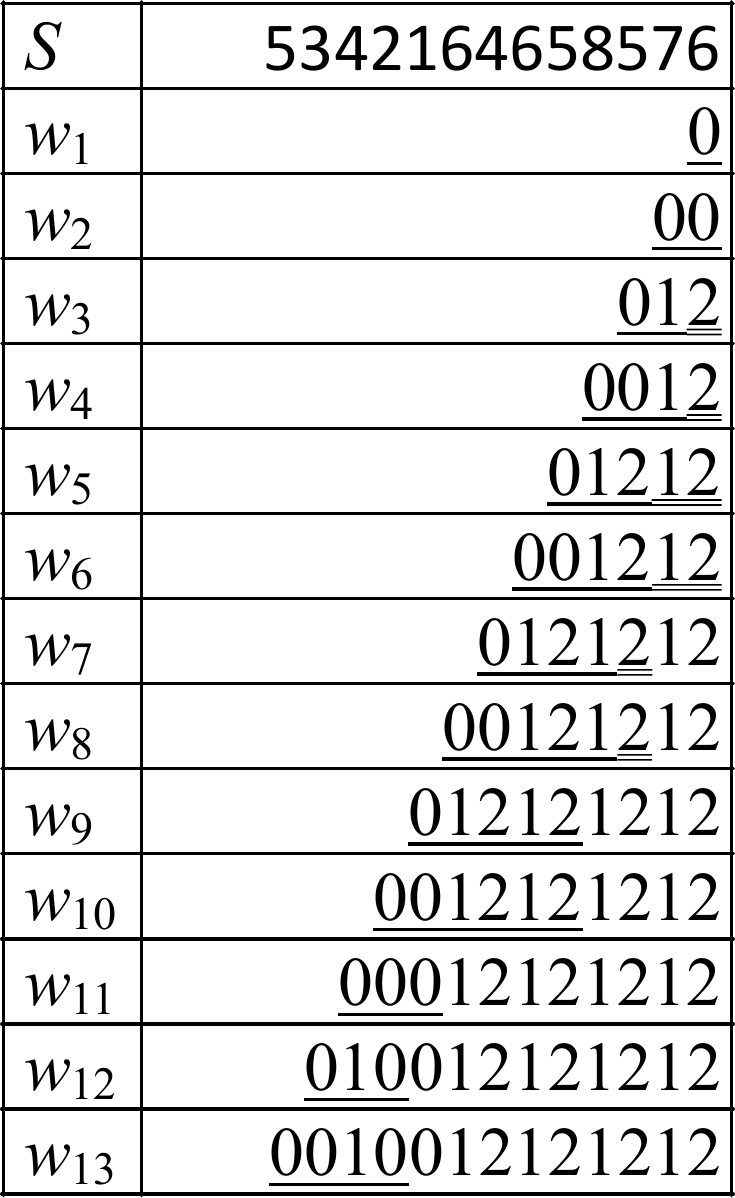}
  }
  \caption{The Cartesian position heap of string $S = \mathtt{264275843647576}$ together with the maximal reach pointers (doubly-lined arcs). For each $w_i = \PD(S[n-i+1..])$, the singly-underlined prefix is the string that is represented by the node $u_i$ in $\CPH(S)$, and the doubly-underlined substring is the string skipped by the maximal reach pointer.}
  \label{fig:CPH_mrp}
\end{figure}

\begin{figure}[!h]
  \centerline{
    \includegraphics[scale=0.4]{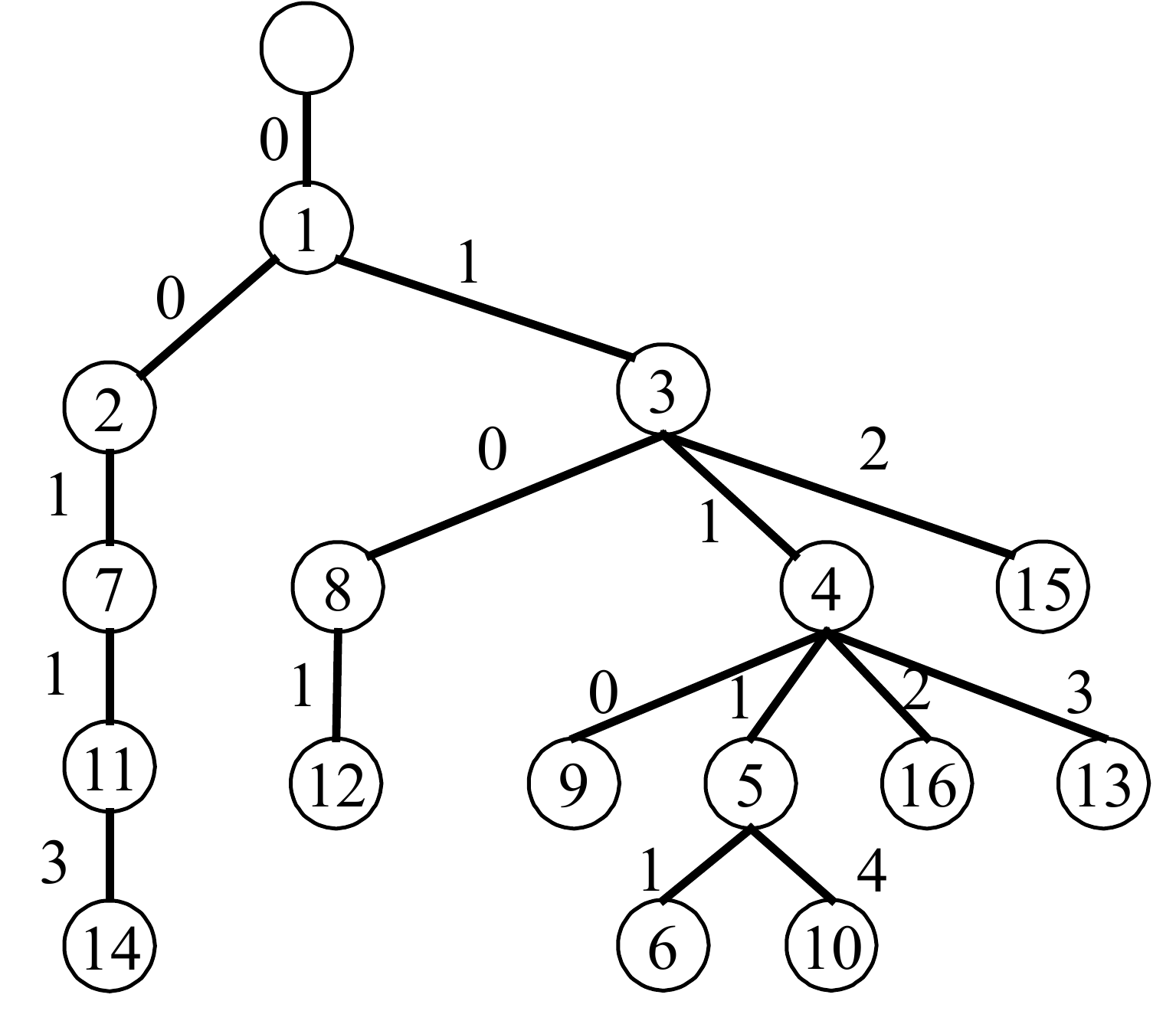}
    \hfill
    \includegraphics[scale=0.38]{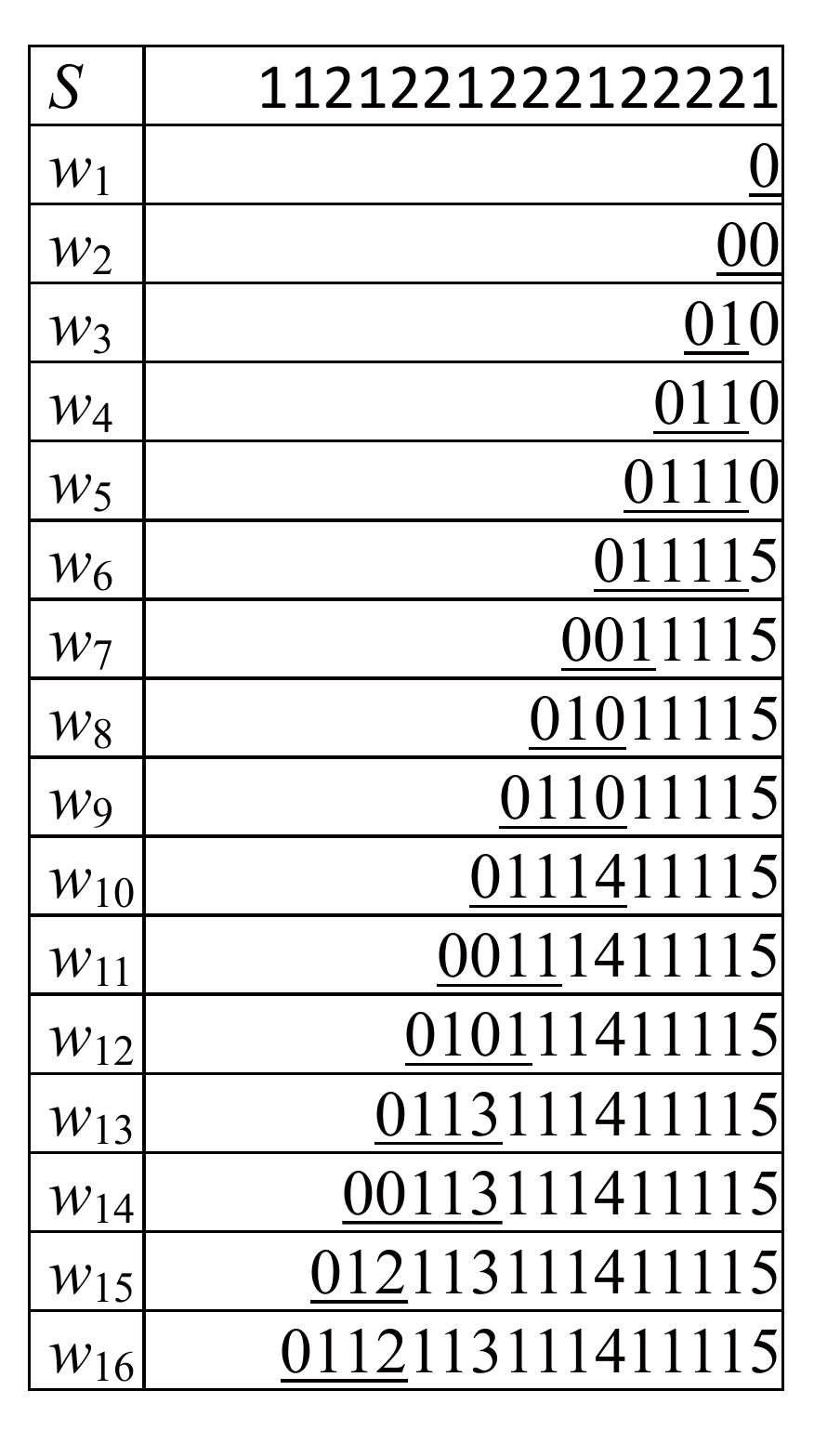}
  }
  \caption{The string $S$ of Lemma~\ref{lem:many_branches} with $k = 4$, i.e., $S = \mathtt{1121221222122221}$ and its Cartesian-tree position heap $\CPH(S)$. For each $w_i = \PD(S[n-i+1..])$, the underlined prefix is represented by the node of $\CPH(S)$ that corresponds to $w_i$. Node $011$ has $k = 4$ out-going edges. This example also shows that the upper bound of Lemma~\ref{lem:number_of_branches} is tight, since node $011$ has $|011|+1 = 4$ out-going edges.}
  \label{fig:lower_bound_branches}
\end{figure}

\begin{figure}[!h]
  \centerline{
    \includegraphics[scale=0.35]{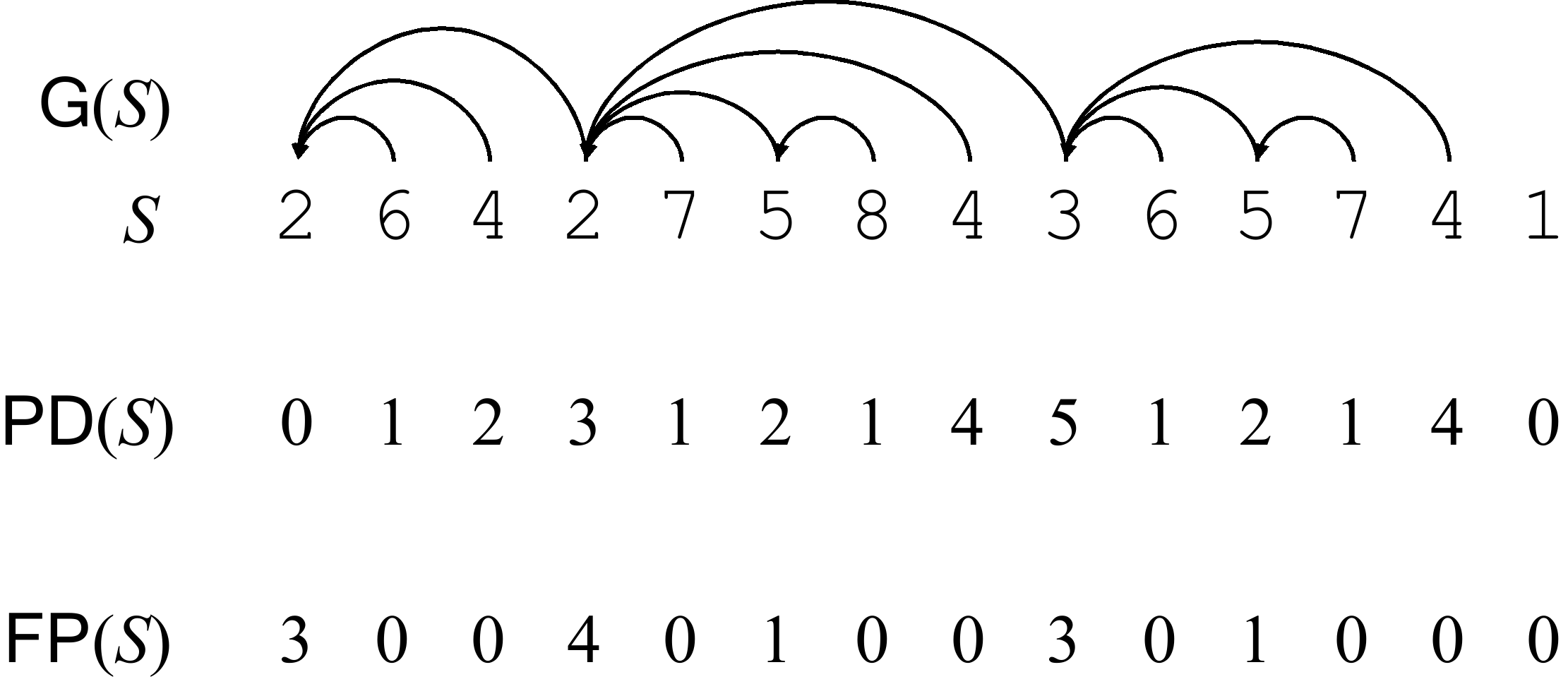}
  }
  \caption{DAG $\mathsf{G}(S)$, $\PD(S)$, and $\FP(S)$
  for string $S = \mathtt{26427584365741}$.}
  \label{fig:DAG}
\end{figure}

\begin{figure}[!h]
  \begin{minipage}{0.69\textwidth}
    \raisebox{10mm}{\includegraphics[scale=0.35]{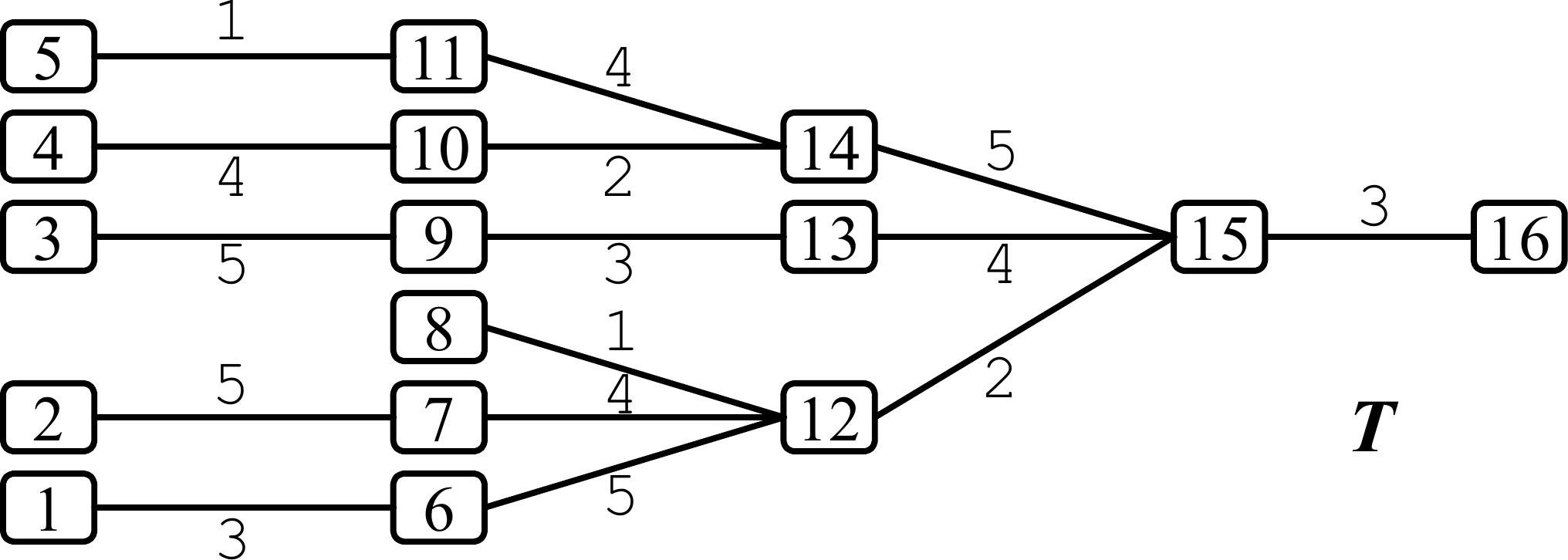}}
    \includegraphics[scale=0.35]{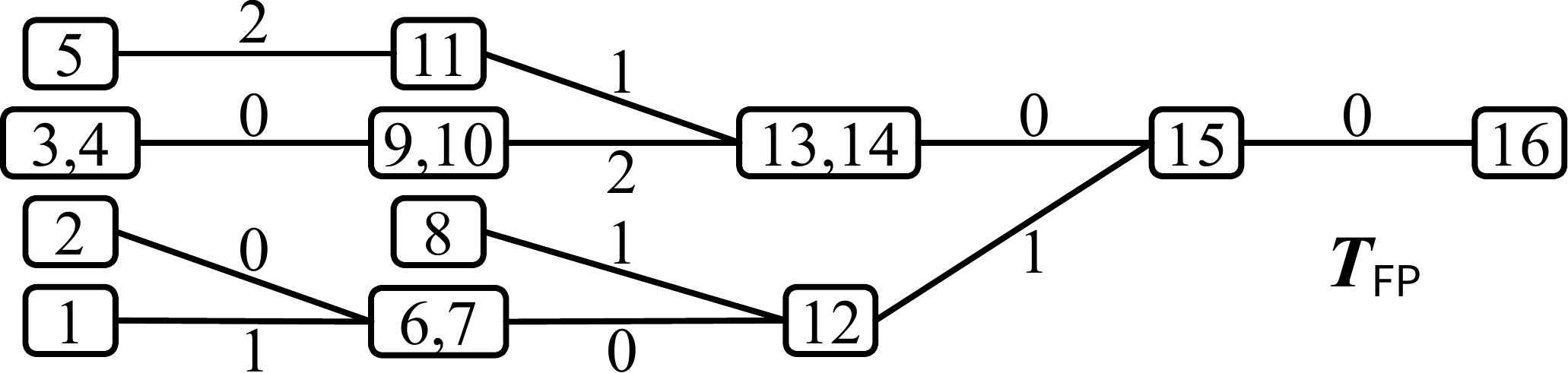}
  \end{minipage}
  \begin{minipage}{0.29\textwidth}
    \includegraphics[scale=0.33]{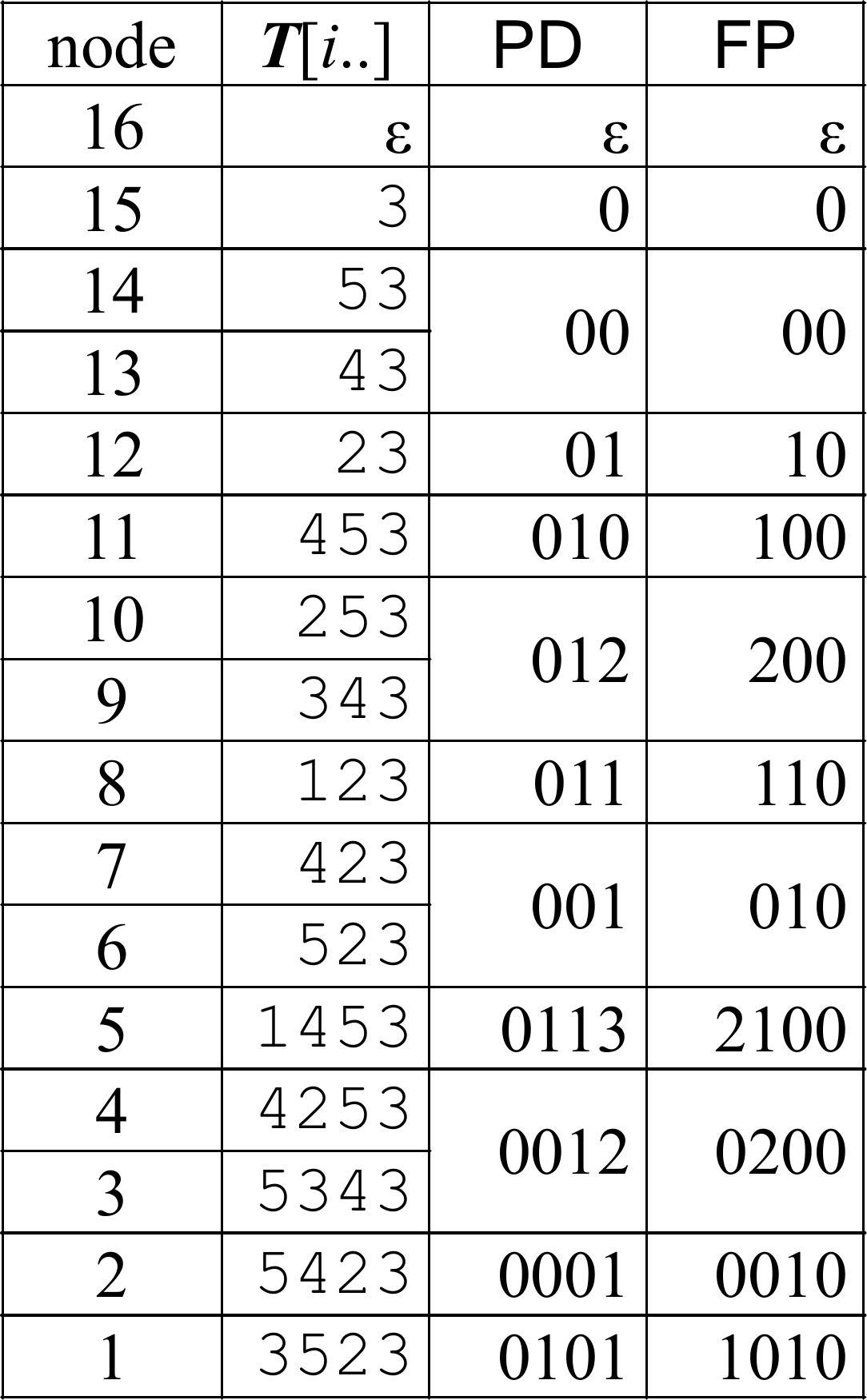}
  \end{minipage}
  \vspace*{2mm}
  \caption{Left upper: An example of input trie $\inputtrie$. Left lower: The FP-trie $\fptrie$ that is obtained from the original trie $\inputtrie$. For instance, the FP encodings of the two path strings $\inputtrie[3..] = \mathtt{5343}$ and $\inputtrie[4..] = \mathtt{4253}$ have the same FP encoding $0200$ and thus the node id's $3$ and $4$ are stored in a single node in $\fptrie$. The representative (the id) of the node $\{3, 4\}$ in $\fptrie$ is $\min\{3, 4\} = 3$.}
  \label{fig:input_trie}
\end{figure}

\begin{figure}[!h]
  \centerline{
    \includegraphics[scale=0.4]{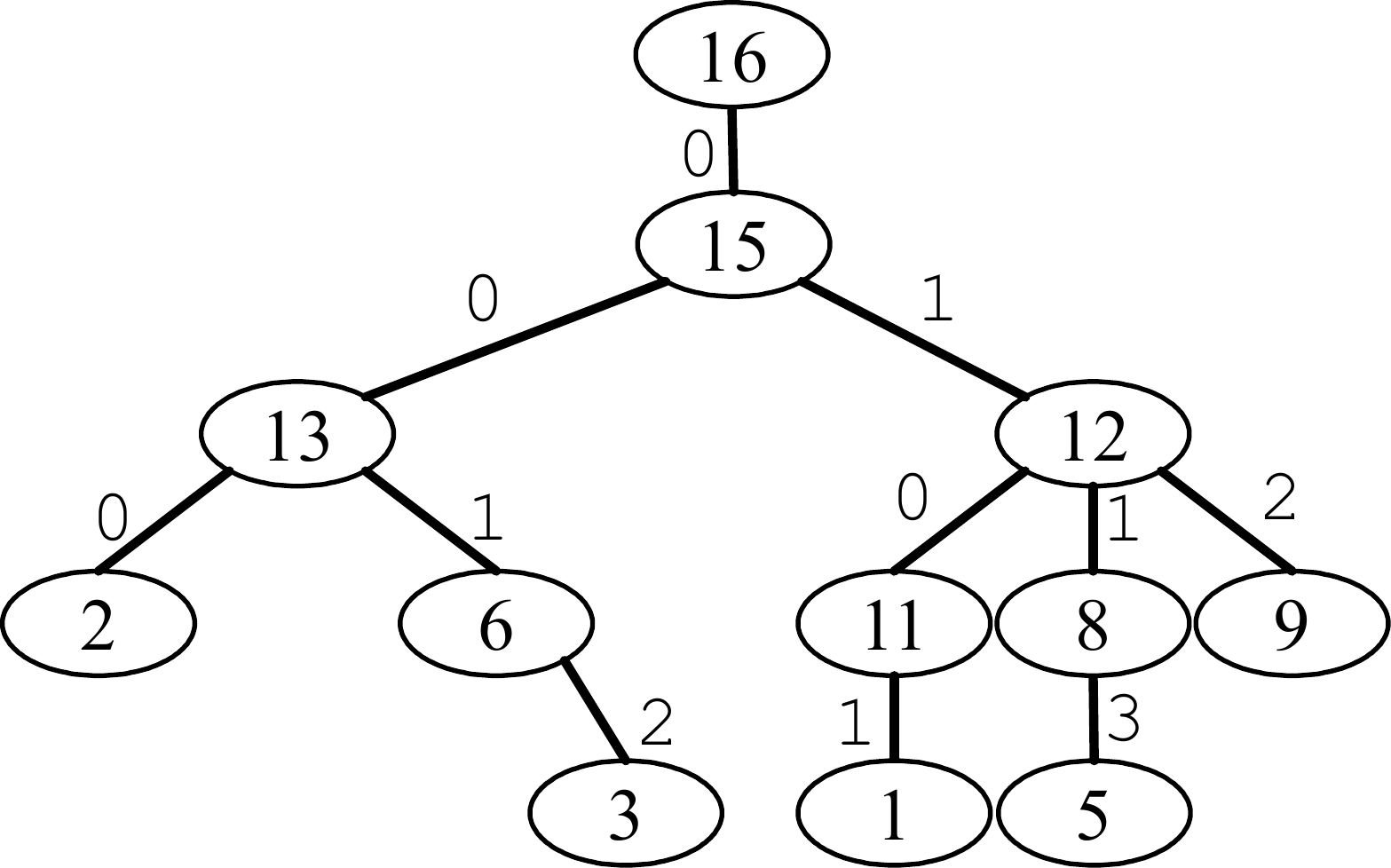}
  }
  \caption{$\CPH(\inputtrie)$ for the trie $\inputtrie$ of Fig.~\ref{fig:input_trie}, where every node $v$ store the representatives $1, 2, 3, 5, 6, 8, 9, 11, 12, 13, 15, 16$ of the corresponding equivalence class $\mathcal{C}_v$.}
  \label{fig:CPM_trie}
\end{figure}

\begin{figure}[!h]
  \centerline{
    \includegraphics[scale=0.4]{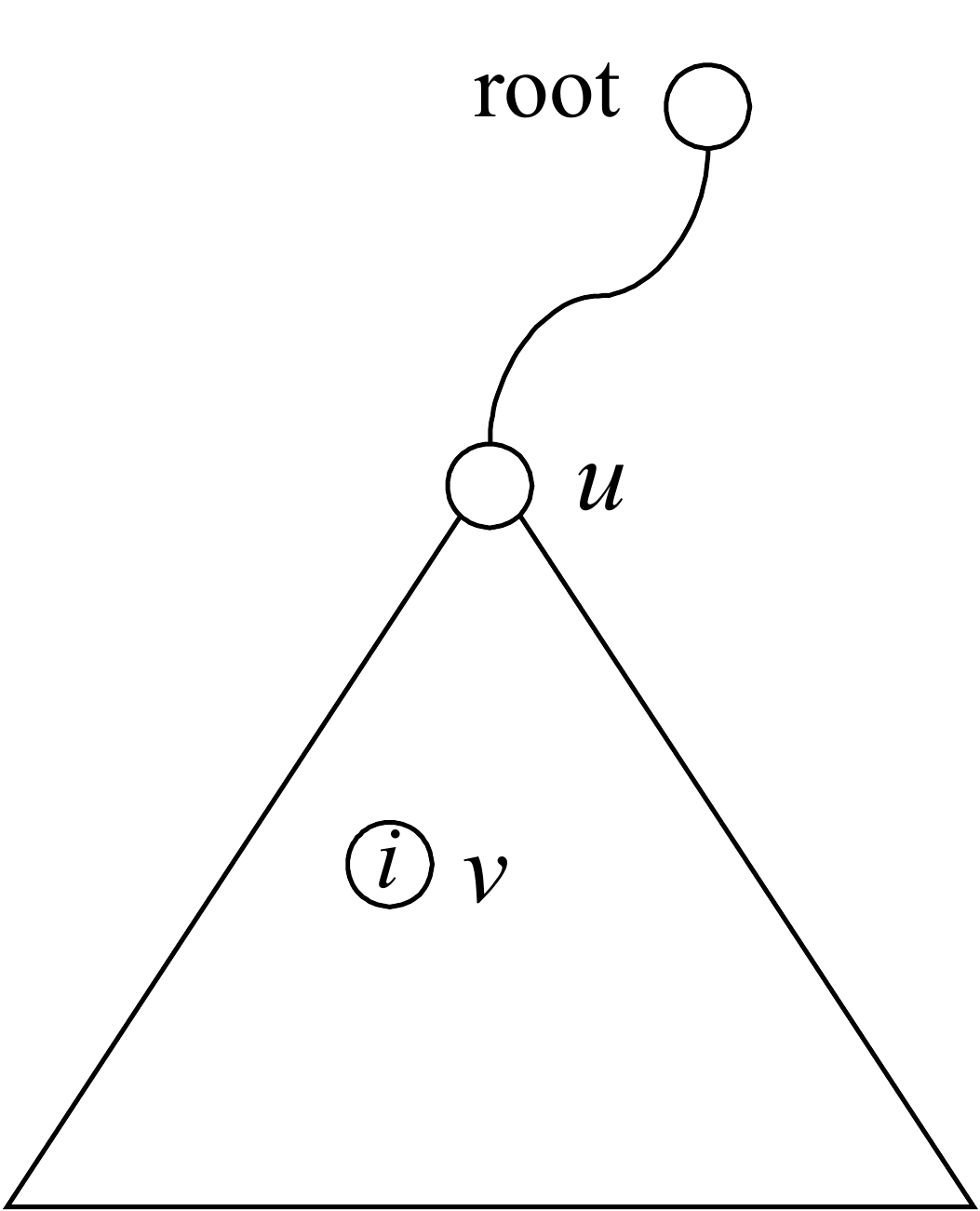}
    \hfill
    \includegraphics[scale=0.4]{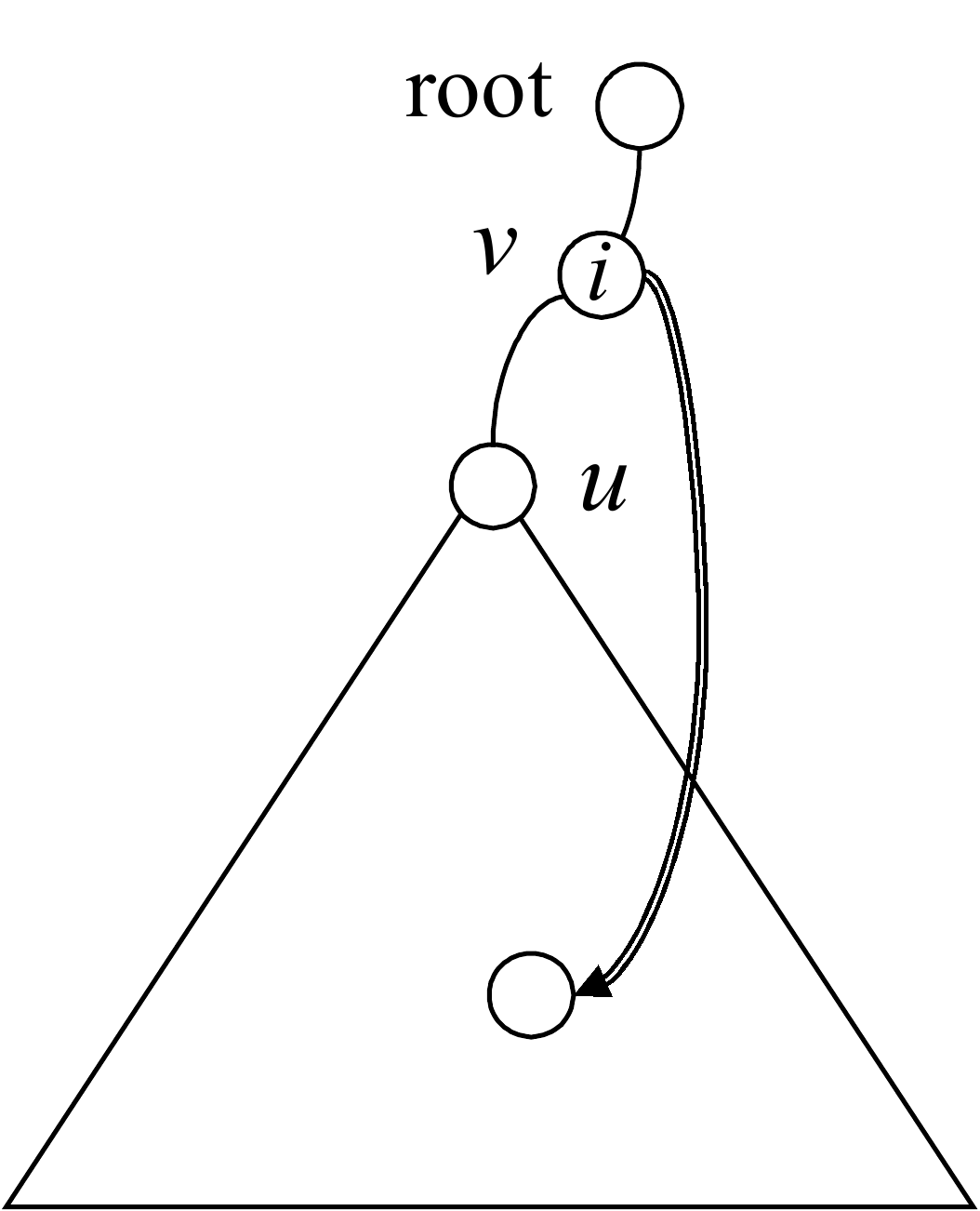}
  }
  \caption{Left: Illustration for Condition (a) of Lemma~\ref{lemma:decendant_condition}. Right: Illustration for Condition (b) of Lemma~\ref{lemma:decendant_condition}, where the doubly-lined arc represents the maximal reach pointer.}
  \label{fig:matching}
\end{figure}

\end{document}